\documentclass[11pt]{article}

\usepackage{sn-preamble}
\usepackage{enumitem}
\usepackage{url}
\usepackage{booktabs}
\usepackage{wasysym}
\usepackage{fancyhdr}
\usepackage{microtype}
\usepackage{pgfplots}
\usepackage{tikz} 
\usepackage{subcaption}
\pgfplotsset{compat=newest}

\newcommand{\Ball}{\mathrm{Ball}}

\newcommand{\dW}{\mathrm{d}_{\mathrm{W}}}
\newcommand{\dDe}{\mathrm{d}_{\mathrm{HH}}}
\newcommand{\dtor}{\mathrm{d}_{\mathrm{tor}}}
\newcommand{\Par}{\mathrm{Par}}
\newcommand{\Emb}{\mathrm{Emb}}
\newcommand{\Prod}{\mathrm{Prod}}
\newcommand{\Mix}{\mathrm{Mix}}
\newcommand{\ed}{\eps_{\mathrm{dist}}}

\newcommand{\Dirac}{\mathrm{Comb}}

\renewcommand{\Re}{\mathrm{Re}}
\renewcommand{\Im}{\mathrm{Im}}

\title{
Model-agnostic distribution recovery in high dimension\\ via Fourier measurements
}

\author{%
 \begin{tabular}{ccc}
     Xi Chen & Anindya De & Yizhi Huang \\
     Columbia University & University of Pennsylvania & Columbia University \\
     \url{xichen@cs.columbia.edu} & \url{anindyad@cis.upenn.edu} & \url{yizhi@cs.columbia.edu} \\
     \\
     Shivam Nadimpalli & Rocco A. Servedio & Tianqi Yang \\
     MIT & Columbia University & Columbia University \\
     \url{shivamn@mit.edu} & \url{rocco@cs.columbia.edu} & \url{tianqi@cs.columbia.edu}
 \end{tabular}
 \vspace{0.5em}
}

\allowdisplaybreaks

\begin{document}

\pagenumbering{gobble}

\maketitle

\begin{abstract}
We study the problem of reconstructing an unknown probability distribution to high accuracy, given  noisy information about its low-degree Fourier coefficients.
Prior results 
for this basic problem have typically imposed strong modeling assumptions on the distribution,
often requiring that it is a linear combination of spatially separated point sources.

In this work we analyze a very general version of the high-dimensional distribution recovery
 problem, by considering completely general 
 distributions over the $d$-dimensional torus $[0,1)^d$; we do not assume any spatial separation between point sources, or even that the distribution is a finite linear combination of point sources.  The question naturally arises:  what can be said about 
 distribution recovery in such a general setting?  
\begin{flushleft}
\begin{itemize}
\item As a warm-up, we first give a
set of results for reconstructing distributions with respect to the Wasserstein distance.  We establish essentially matching upper and lower bounds on the cutoff frequency $T$ and the magnitude $\kappa$ of the noise for which accurate reconstruction  is possible. Roughly speaking, our results here show that for $d$-dimensional distributions, estimates of $\approx \exp(d)$ many Fourier coefficients are both necessary and sufficient for accurate reconstruction under the Wasserstein distance. 
\item  As our main result, we define a new notion of ``heavy hitter'' reconstruction for distributions, which 
essentially amounts to achieving high-accuracy reconstruction of all ``sufficiently dense'' regions of the distribution.  
We give essentially matching upper and lower bounds on the cutoff frequency $T$ and the magnitude $\kappa$ of the noise for which accurate reconstruction is possible under this notion.
Our results show that\hspace{0.1cm}---\hspace{0.1cm}in sharp contrast with Wasserstein reconstruction\hspace{0.1cm}---\hspace{0.1cm}accurate estimates of only $\approx \exp(\sqrt{d})$ many Fourier coefficients are both necessary and sufficient for heavy hitter reconstruction.
\end{itemize}
\end{flushleft}
\end{abstract} 

\newpage  


\newpage

\pagenumbering{arabic}

\section{Introduction} \label{sec:intro}


This paper takes a theoretical computer science perspective on the fundamental 
%
problem of recovering an unknown probability distribution from coarse, noisy Fourier measurements.
Many prior works on variants of this problem, in signal processing and related research communities, typically relied on rather stringent assumptions on the unknown distribution such as that it is a finite combination of spatially separated point sources.
In contrast, we take a maximally broad 
view of the problem by considering 
general high-dimensional probability distributions 
\emph{without imposing any structural modeling assumptions}.  

Tackling the distribution recovery problem at this level of generality requires us to carefully consider the precise notion of ``recovery'' that can be achieved. A conceptual contribution of this work is that it introduces a new recovery notion for high-dimensional distributions, which we call \emph{heavy hitter} reconstruction.  As our main technical contribution, we show that using this new notion it is possible to meaningfully evade the usual ``curse of dimensionality'' and achieve information-theoretic reconstruction with a number of measurements that scales much more favorably than exponential in the dimension parameter $d$ (and we give a matching lower bound).

Before stating our  results in \Cref{subsec:our-results}, we briefly review relevant prior work on distribution recovery
which will provide context for our formulation relative to the classical settings considered in the literature.

\subsection{Background on Fourier-based distribution recovery
}

Every probability distribution is completely specified by its Fourier transform, or equivalently, its characteristic function~\cite[Theorem 3.3.11]{durrett2019probability}.
This fundamental fact, which traces back at least to the work of Fourier and de Moivre, makes Fourier measurements a natural input for distribution recovery.  Of course, it is not feasible to provide a recovery algorithm with the entire (infinite) set of Fourier coefficients. Thus 
the challenge, and the focus of much recent work, is to obtain finitary and robust versions of this reconstruction principle: i.e., to understand how much noisy, low-frequency Fourier information suffices to recover an unknown distribution in a meaningful sense.

Although our focus will be on high-dimensional distributions, it is instructive to begin with the one-dimensional setting that was studied by Cand\`{e}s and Fernandez-Granda~\cite{candes2013super}. 
Here, the domain is the one-dimensional torus $\mathbb{T}$, which we identify with the interval $[0,1)$, and the object to be recovered is 
a discrete distribution $D$ 
over $\mathbb{T}$ that is supported on $k$ points, which we write as 
\begin{equation} \label{eq:intro-our-signal-model}
    D(x) := \sum_{j=1}^k u_j \delta_{t_j}(x)\,,
\end{equation}
where $\delta_{t_j}(\cdot)$ denotes a Dirac delta centered at the point $t_j \in \mathbb{T}$ and $u_j \in \R_{>0}$  is the corresponding probability (see \Cref{fig:illustration}(a) for an illustration). 
In the signal processing literature and the literature on ``super-resolution'' such a $D$ is often referred to as a superposition of $k$ \emph{point sources}, where the $\smash{j^\text{th}}$ point source has location $t_j$ and amplitude $u_j$. 
We will sometimes refer to a Dirac delta $\delta_t$ as  a \emph{spike} located at $t$. 

Next, recall that the Fourier spectrum over the torus $\mathbb{T}$ is indexed by the integers $\Z$. 
In particular, for a distribution $D$, 
the Fourier coefficient at frequency $\ell \in \Z$ is given by 
\[
    \widehat{D}(\ell) = \int_{x \in  \mathbb{T}} D(x) \cdot \exp\big( 2 \pi i \ell x\big) \, dx\,.
\]
Given that $\int_x g(x) \delta_t(x) dx = g(t)$ for any test function $g$, for $D$ given in \Cref{eq:intro-our-signal-model} we have 
\[
    \wh{D}(\ell) = \sum_{j = 1}^k u_j \exp\big(2\pi i \ell t_j\big).
\]

The main question posed by Cand\`es and Fernandez-Granda in \cite{candes2014towards} was whether bandlimited measurements suffice to recover an unknown distribution $D$, under the promise that $D$ is a superposition of $k$ point sources (i.e., $D$ is as in \Cref{eq:intro-our-signal-model}). 
They proposed a convex optimization framework for this inverse problem and showed, among other things, that if $\wh{D}(\ell)$ is known \emph{exactly} for all $|\ell| \leq k$ (in other words, the distribution is observed up to bandlimit $k$), then their convex program perfectly reconstructs $D$.

It is worth noting that the ability to recover a $k$-sparse distribution $D$ from the \emph{exact} Fourier coefficients \smash{$\{\wh{D}(\ell)\}_{|\ell| \leq k}$} is not new: Prony's method, dating back to 1795~\cite{prony1795essai}, achieves the same guarantee. 
The novelty of \cite{candes2014towards} lies in showing that such recovery can be accomplished through a convex optimization program. 
That said, the assumption of having perfect access to the \emph{exact} Fourier data is quite restrictive; in practice, the measurements $\wh{D}(\ell)$ are inevitably perturbed by noise.  
This raises the natural question of whether recovery is possible given \emph{noisy measurements.} \medskip

\noindent\textbf{Noisy measurements.} 
Suppose the algorithm no longer has access to the true Fourier coefficients but instead receives noisy estimates $\{\widetilde{D}(\ell)\}_{|\ell| \le T}$ satisfying 
\begin{equation}\label{eq:hehe1}
    \big|\widetilde{D}(\ell) - \widehat{D}(\ell)\big| \le \kappa,\quad\text{for every $\ell$ with $|\ell|\le T$.}
\end{equation}
In other words, the observed coefficients are perturbed by additive noise of magnitude at most $\kappa$. 
This modification substantially changes the nature of the problem. 
Indeed, it is easy to see that given $\{\widetilde{D}(\ell)\}_{|\ell| \le T}$, one cannot distinguish between 
\begin{equation} \label{eq:ithaca}
    D_1(x) = \delta_0(x) 
    \quad\text{versus}\quad 
    D_2(x) = \frac{1}{2}\cdot\delta_0(x) + \frac{1}{2}\cdot\delta_{\Delta}(x)
\end{equation}
for $\Delta \ll \min\{\kappa, T^{-1}\}$. 
In other words, even determining whether $D$ consists of one spike or two becomes impossible without imposing an additional \emph{separation} condition between the spikes.

Thus, to recover a $k$-sparse distribution $D$ from noisy estimates of $\wh{D}(\ell)$, it appears necessary to impose a separation assumption on the support points of $D$ (spoiler alert: this is precisely the assumption we re-examine in this paper). 
To formalize what ``separation'' means on the torus, we first recall the natural distance metric between two points $a, b \in \mathbb{T}$: the \emph{toroidal distance} $\dtor(a,b)$, which is the shortest distance between $a$ and $b$ when identifying the interval $[0,1)$ with the unit circle (so that $0$ and $1$ coincide; see \Cref{sec:fourier-torus} for the formal definition). 

With this definition in hand, Cand\`es and Fernandez-Granda~\cite{candes2013super} studied the setting where any two supports points $t_i, t_j$ of $D$ satisfy $\dtor(t_i, t_j) \geq \Delta$. 
They proved that $D$ can be approximately recovered even in the presence of noise, provided that the support points $t_1, \dots, t_k$ lie on a grid and the bandlimit $T$ is at least $2/\Delta$. 
(We note that the error metric in \cite{candes2013super} is based on the Fej\'{e}r kernel and is somewhat nonstandard.) 
Subsequent work~\cite{supportdetection2013,liao2016music} removed the grid assumption, although the resulting noise-tolerance guarantees were not directly comparable to those of \cite{candes2013super}.

Before proceeding, it will be helpful to standardize some notation and terminology that will be used throughout our discussion: 
\begin{flushleft}\begin{enumerate}
    \item We will denote by $T$ the \emph{bandlimit}; that is, the maximum frequency index (or degree of the Fourier coefficients) available to the algorithm. 
    \item We will denote by $\kappa$ the \emph{noise level per Fourier coefficient}; that is, for all $\ell$ the true Fourier coefficient $\wh{D}(\ell)$ and its noisy estimate $\wt{D}(\ell)$ satisfy \Cref{eq:hehe1}.
    \item If $D$ is supported on a finite set of points, then we will write $k$ for its \emph{sparsity} and use
    \[
        \Delta := \min_{i,j \in [k]} \dtor(t_i, t_j)
    \]
    to denote the \emph{minimum separation} between them. 
\end{enumerate}\end{flushleft}

The most relevant prior work in the one-dimensional setting is that of Moitra~\cite{moitra2015super}, which was the first to study problems of this kind from the perspective of theoretical computer science. 
Under the promise that $D$ is a discrete distribution with minimum separation at least
  $\Delta$, Moitra showed that as long as $T$ and $\kappa$ satisfy 
\[
    T \geq \frac{1}{\Delta} + 1 \quad\text{and}\quad \kappa \leq \poly\pbra{\frac{1}{\Delta}, \eps}\,,
\]
$D$ can be efficiently recovered to error $\eps$; that is, each support location $t_i$ and amplitude $u_i$ can be recovered to within additive error $\eps$. 
Note that, while no sparsity bound is explicitly assumed on $D$,  a minimum separation of $\Delta$ implies that the sparsity $k$ of $D$ satisfies $k \le 1/\Delta$ on the one dimensional torus.

\begin{figure}[t]
\centering

\begin{subfigure}{\linewidth}
\centering
\begin{tikzpicture}
\begin{axis}[
    width=13cm,
    height=5cm,
    xlabel={$x$},
    ylabel={Probability mass},
    xmin=0, xmax=1,
    ymin=0, ymax=1.3,
    xtick={0,1},                
    yticklabels={},
    axis line style={black},
    tick style={draw=none},
    xtick={0,1}, 				 
    yticklabels={},              
    grid=none,
    title style={yshift=2ex},
	legend style={at={(0.5,1.02)}, anchor=south,
	    draw=black,              
        fill=white,
        legend columns=-1,
        minimum width=2.388cm,    
        }
]

\addplot+[ycomb, thick, black, mark=*, mark options={fill=black}, mark size=1pt] 
  coordinates {(0.2,1.0) (0.5,0.8) (0.8,1.1)};
\addlegendentry{\small True distribution}

\addplot+[ycomb, very thick, red, mark=*, mark options={fill=red}, mark size=1pt]
  coordinates {(0.18,0.9) (0.52,0.7) (0.82,1.0)};
\addlegendentry{\small Recovered distribution}

\end{axis}
\end{tikzpicture}
\caption{An illustration of classical 
distribution recovery as considered in prior works~\cite{candes2013super,moitra2015super}, which assumes that the underlying distribution is a discrete mixture of Dirac delta $\delta_t(\cdot)$ point masses.
}

\end{subfigure}

\vspace{1.5em}

\begin{subfigure}{\linewidth}
\centering
\begin{tikzpicture}
\begin{axis}[
    width=13cm,
    height=5cm,
    xlabel={$x$},
    ylabel={Probability mass},
    xmin=0, xmax=1,
    ymin=0, ymax=1.2,
    grid=none,
    tick style={draw=none},
    xtick={0,1}, 				 
    yticklabels={},              
	legend style={at={(0.5,1.02)}, anchor=south,
	    draw=black,              
        fill=white,
        legend columns=-1,
        minimum width=2.388cm,    
        }
]

\addplot[
    thick, black!80!white, domain=0:1, samples=400
]
{ exp(-((x-0.25)/0.025)^2)
 + 0.8*exp(-((x-0.55)/0.03)^2)
 + 0.4*exp(-((x-0.8)/0.04)^2)
};
\addlegendentry{\small True distribution\,}

\addplot[
    thick, blue, dotted, domain=0:1, samples=500
]
{ 0.88*exp(-((x-0.24)/0.029)^2)   
 + 0.7*exp(-((x-0.56)/0.04)^2)   
 + 0.35*exp(-((x-0.82)/0.05)^2)  
};
\addlegendentry{\small Wasserstein recovery\,\,}

\addplot[
    thick, red, densely dashed, domain=0:1, samples=500
]
{ 1.1*exp(-((x-0.26)/0.02)^2)
 + 0.87*exp(-((x-0.54)/0.025)^2)
};
\addlegendentry{\small Heavy hitter recovery}

\addplot[dashed,gray!80!white]
coordinates {(0,0.6) (1,0.6)};
\end{axis}

\end{tikzpicture}
\caption{An illustration of \emph{model-agnostic} 
distribution recovery
considered in this work, where the distribution may be arbitrary (discrete or continuous). 
The gray line indicates the recovery threshold for ``heavy hitters.''}

\end{subfigure}

\caption{An illustration of the modeling assumptions and recovery objectives in this work compared to prior formulations.}
\label{fig:illustration}

\end{figure}
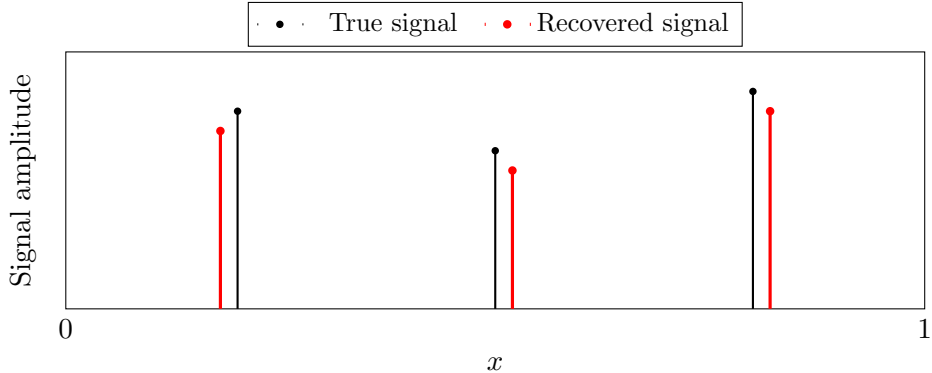
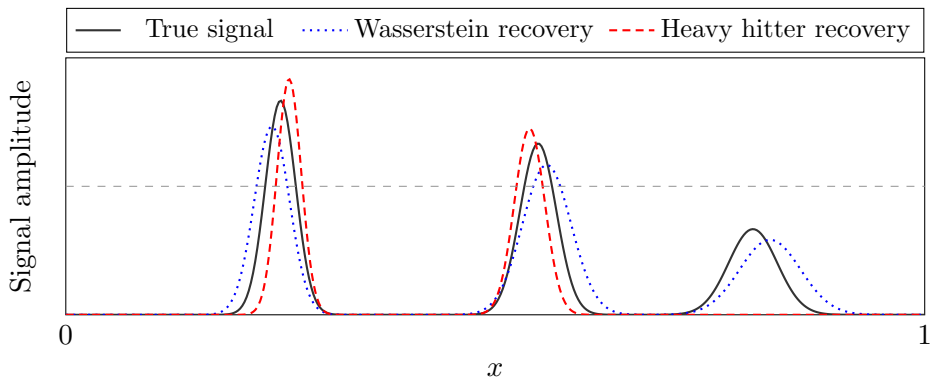

\subsection{This work: Recovery \emph{without} sparsity or separation \emph{in high dimension}} The preceding work, including \cite{moitra2015super}, naturally raises the following question: 
\begin{quote}
\begin{center}
Can we still recover $D$ \emph{without} any separation or sparsity assumptions?
\end{center}
\end{quote} 
In particular, we eschew the assumption that $D$ is a mixture of finitely (or even countably) many point sources. 
Instead, we allow $D$ to be an arbitrary probability measure on $\mathbb{T}$ of total mass $1$. 
For intuition, the reader may think of $D$ as having a density $D:\mathbb{T}\to \R_{\ge 0}$ satisfying
\[
    \int_{x\in \mathbb{T}} D(x)\,dx = 1\,.
\]
(Strictly speaking, this is not accurate in the presence of Dirac delta masses which do not admit densities; however, we will freely adopt this viewpoint in the following informal discussion as it captures the relevant intuition. See \Cref{sec:prelims} for a more formal treatment.)
To connect this with the previous discussion, observe that when $D (x)= \sum_{j=1}^k u_j \delta_{t_j}(x)$ as in \eqref{eq:intro-our-signal-model}, we have 
\[
    \int_{x \in \mathbb{T}} D(x)\, dx = \sum_{j=1}^k u_j\,.
\]

Of course, in such a general setting it is no longer meaningful to speak of recovering the support of the distribution $D$. 
Moreover, while the discussion thus far has centered on the one-dimensional case, the aim 
of this paper is to develop a modeling-agnostic theory of 
distribution recovery
in \emph{high dimensions}: namely, recovering distributions over the $d$-dimensional torus $[0,1)^d$ from noisy Fourier measurements. 
(Looking  ahead, a central goal in this effort will be to achieve such recovery with a complexity that scales better than exponentially with the dimension $d$.)

A first natural notion of recovery in this setting is recovery in \emph{Wasserstein} (or \emph{earthmover}) distance, 
and indeed as a warm-up to our main results we give near-matching positive and negative results for  Wasserstein distance recovery,
though necessarily with bounds that grow exponentially in $d$. 
(As we discuss more below,
recent work by Musco, Musco, Rosenblatt, and Singh~\cite{MMRS25} obtains  guarantees that are similar 
to our warm-up positive results for Wasserstein distance recovery, though using moment information rather than Fourier measurements.)
Thus, while Wasserstein recovery allows one to move beyond the sparsity assumption, it succumbs to
the curse of dimensionality. 
This motivates us to consider a 
different notion of recovery, which is weaker than Wasserstein distance, but still captures the intuition of identifying the significant ``spikes'' in the data. 
Developing this notion, which we call \emph{heavy hitter recovery}, and showing that it avoids the full exponential-in-dimension dependence inherent in Wasserstein recovery is the main conceptual and technical contribution of this paper.

We briefly survey additional relevant prior 
work before turning to a more detailed description of our results.
\medskip


\noindent\textbf{Related work.} 
While several earlier works have also studied 
distribution recovery
under the Wasserstein distance or related metrics~\cite{denoyelle2017support,schiebinger2018superresolution,eftekhari2021sparse, fan2023efficient,kurmanbek2023multivariate}, there are several key differences from our work. 
First, although these works do not assume a separation condition, they all assume a sparsity condition\hspace{0.1cm}---\hspace{0.1cm}that the signal is a linear combination of finitely many point sources\hspace{0.1cm}---\hspace{0.1cm}whereas we make no such assumption.
Second, each of the above papers either restricts attention to the one-dimensional case (while our main results are high-dimensional) or adopts a different measurement model, such as approximate moments, instead of the approximate Fourier coefficients used in our work. 

The problem of 
distribution recovery
has also been studied in higher dimensions, for example in \cite{huang2015super,poon2019multidimensional,Liu_Ammari_2024}. 
Another
related line of work studies the
recovery of Fourier-sparse signals in the continuous setting~\cite{chen2016fourier,chen2019estimating,song2023quartic}. 
However, these works fundamentally rely on structural assumptions on the signal\hspace{0.1cm}---\hspace{0.1cm}most notably sparsity, and 
often
also minimum separation conditions.
In contrast, our goal is to study 
distribution recovery
in a \emph{model-agnostic} manner: we consider completely general  non-negative  signals, without imposing either a sparsity or a separation condition. 

We note in particular the recent work of \cite{MMRS25}, which, like ours, studies the recovery of general (not necessarily sparse) signals in high dimensions, but considers approximate moments as the measurement model rather than Fourier coefficients. 
In particular, our positive result for Wasserstein distance recovery is quite similar to Theorem~32 of~\cite{MMRS25}, though we obtain a modest quantitative improvement in the dependence on the dimension  
(see \Cref{rem:MMRS} for further discussion). 
We complement this positive result with near-matching lower bounds, and, more importantly, introduce and analyze the notion of \emph{heavy hitter recovery}, which yields a separation from Wasserstein recovery and circumvents the so-called ``curse of dimensionality.'' 

\subsection{Our results}
\label{subsec:our-results}

As mentioned above, the main goal of this paper is to develop techniques for 
distribution recovery
without assumptions in high dimensions.
Our domain of interest is the $d$-dimensional torus $\mathbb{T}^d$ which we identify with $[0,1)^d$. 
As before, a distribution $D$ over $\mathbb{T}^d$ will be normalized to satisfy 
\[
    \int_{x\in\mathbb{T}^d} D(x) \, dx = 1\,.
\]
The Fourier coefficients of $D$ are now indexed by integer tuples $\ell \in \Z^d$ and are defined as
\[
    \widehat{D}(\ell) := \int_{x \in \mathbb{T}^d} D(x) \cdot \exp\big(2\pi i \ell \cdot x\big) \, dx
\]
where $\ell\cdot x$ denotes the usual inner product in $\R^d$. 
For two points $x,y \in \mathbb{T}^d$ the toroidal distance $\dtor(x,y)$ now corresponds to 
\[
    \dtor(x,y):=\sqrt{\sum_{i=1}^d \dtor^2(x_i, y_i)}\,,
\]
that is, the Euclidean distance between $x$ and $y$ where each coordinate distance is measured using the one-dimensional toroidal metric. 

\subsubsection{Wasserstein distance recovery} 

We first describe our Wasserstein distance recovery results, which serve as a warm-up to our heavy hitter distance recovery results. 
Given two distributions $D_1$ and $D_2$, recall that their \emph{Wasserstein distance}, denoted by $\dW(D_1, D_2)$, is defined as
\[
    \dW(D_1, D_2) := \sup_{g: \mathbb{T}^d \to \mathbb{R}} \int g(x)\cdot (D_1(x) - D_2(x))\,dx,
\]
where the supremum is taken over all $1$-Lipschitz functions $g$ on $\mathbb{T}^d$ (with respect to the toroidal distance).

Our positive result for Wasserstein distance recovery, \Cref{thm:multi-dim-Wass-ub}, can be summarized as follows: For any $\epsilon>0$, given noisy Fourier coefficients $\{\widetilde{D}(\ell)\}_{\Vert \ell \Vert_\infty \le T}$ for an unknown distribution $D$ with\vspace{-0.1cm}
\[
    T = \Theta\pbra{\frac{\sqrt{d}}{\eps}}
    \quad\text{and}\quad
    \kappa = \Theta\pbra{\frac{\eps}{\sqrt{d}}}^d,\vspace{-0.05cm}
\]
we can recover $D$ up to Wasserstein distance $\epsilon$. 
This result bears both technical and conceptual similarities to Theorem~32 of \cite{MMRS25}; see \Cref{rem:MMRS} for further discussion. 
Moreover, for constant dimension $d$, we give an algorithm that runs in time polynomial in its input size 
(see \Cref{sec:wasserstein-algo}).

Note that the admissible noise level $\kappa$ decays exponentially with $d$, and the number of Fourier coefficients with $\Vert \ell\Vert_\infty \le T$\hspace{0.1cm}---\hspace{0.1cm}and hence the number of measurements\hspace{0.1cm}---\hspace{0.1cm}also grows exponentially in $d$. 
These dependencies turn out to be unavoidable:
\begin{flushleft}
\begin{itemize}
	\item Our first lower bound, \Cref{thm:low-dim-kappa-zero}, shows that recovering $D$ up to Wasserstein error $\epsilon$ requires Fourier coefficients up to bandlimit $T = \Omega(\sqrt{d}/(2\epsilon))$ even in the noiseless ($\kappa = 0$) case. 
	\item One might hope to trade off bandlimit and noise, but our second lower bound, \Cref{thm:lb-infinite}, rules this out: even if all Fourier coefficients are available ($T = \infty$), achieving Wasserstein error $\epsilon$ still requires the noise level $\kappa$ to be at most $O(\epsilon^{d/4})$. 
\end{itemize} 
\end{flushleft}
Taken together, these results give nearly tight upper and lower bounds on the tradeoffs between bandlimit and noise. 
They also reveal a fundamental limitation: Wasserstein recovery in high dimensions necessarily incurs an exponential dependence on $d$.  
This barrier motivates the main contribution of this paper, which is to consider a weaker notion of recovery that still captures the salient features of the distribution while circumventing the curse of dimensionality. 

\subsubsection{Heavy hitter recovery}

The exponential barrier for Wasserstein recovery is not an artifact of our techniques but a consequence of the recovery goal itself. 
Our lower-bound constructions (\Cref{thm:low-dim-kappa-zero,thm:lb-infinite}) show that the hard instances are distributions whose mass is spread \emph{diffusely} across the domain: no small collection of localized regions carries a significant fraction of the mass, so noisy low-frequency information cannot resolve the global distributional structure except at exponential-in-$d$ cost. 
Moreover, as discussed in \Cref{rem:LP,rem:LP2}, the same constructions also rule out recovery in the L\'evy--Prokhorov metric, another ``global'' metric on distributions. 

This motivates relaxing global reconstruction to the weaker but still meaningful objective of recovering only the  ``significant'' localized parts of the distribution. 
Such a ``heavy hitter'' objective has several natural precedents throughout theoretical computer science: for example, in streaming and sketching, where heavy-hitter guarantees recover the large coordinates of an unknown vector~\cite{charikar2002finding,cormode2005improved,berinde2009space,price2011epsilon,woodruff2016new}; in population recovery, where the goal is to identify the high-probability strings of an unknown distribution over a discrete domain~\cite{moitra2013lossy,polyanskiy2017sample,de2020sharp}; and in list-decodable learning, where the goal is to  recover a significant structured component in the presence of an arbitrary remainder~\cite{charikar2017learning,diakonikolas2018list,CCKL26}.  

Motivated by these works, we introduce the notion of the \emph{heavy hitter distance} between distributions. 
This metric compares distributions by the mass they place on localized high-density regions, rather than by their full global structure.
Since the goal of this work is to be model-agnostic, we continue to make no assumptions about the support of $D$; in particular, $D$ need not have finite or even countable support. \medskip


\noindent\textbf{The heavy hitter distance.}
We will use the following notation: 
for any $x \in \mathbb{T}^d$ and $\tau \geq 0$, define 
\[\Ball(x, \tau) := \big\{y: \dtor(x, y) \le \tau\big\}.\]
For any distribution $D$ over $\mathbb{T}^d$, we define
\[
D\big(\Ball(x,\tau)\big) := \int_{\dtor(x, y) \le \tau} D(y) dy,
\]
i.e., $D(\Ball(x,\tau))$
is the probability that $D$ puts in the toroidal-distance ball of radius $\tau$ around $x$.

The heavy hitter distance between two distributions is parameterized by  a scale parameter $\ed>0$, and is defined as follows:

\begin{definition} [Heavy hitter distance]
\label{def:De-distance}
Given a scale parameter $0 < \ed < 1$, 
the \emph{heavy hitter distance at scale $\ed$} between two distributions $D_1,D_2$ over $\mathbb{T}^d$, written as $\smash{\dDe^{(\ed)}(D_1,D_2)}$, is the smallest value $\eps \geq 0$ such that the following holds:  For every $0 \leq \tau \leq \ed$ 
and every $x \in \mathbb{T}^d$, 
\begin{align*}
D_1\big(\Ball(x,\tau)\big) &\le D_2\big(\Ball(x,\tau + \ed)\big)  + \epsilon\,, \quad \text{and}\\[0.6ex]
D_2\big(\Ball(x,\tau)\big) &\le D_1\big(\Ball(x,\tau + \ed)\big)  + \epsilon\,. 
\end{align*}
\end{definition}

We remark that the term ``heavy hitter distance'' should not be interpreted literally as meaning that $\smash{\dDe^{(\ed)}(\cdot,\cdot)}$ satisfies the triangle inequality; rather, as we discuss later, it is  a way of measuring the error between two distributions vis-a-vis the ``heavy hitter-like'' regions of the distributions. 
Intuitively, having $\smash{\dDe^{(\ed)}(D_1, D_2) \le \epsilon}$ means that for any ``small'' ball, the mass that $D_1$ and $D_2$ put in the ball are essentially the same, provided that we are allowed to dilate the ball by an additive $\ed$ and adjust the probability mass additively by $\epsilon$.

\begin{remark}
    [Heavy hitter distance versus Wasserstein distance]
    \label{rem:HH-Wasserstein}
It is not difficult to show that having small heavy hitter distance is essentially a weaker condition than having small Wasserstein distance\hspace{0.1cm}---\hspace{0.1cm}more precisely, having small Wasserstein distance between two distributions $D_1,D_2$ implies that the heavy hitter distance (at a suitable scale $\ed$) must also be small.  See \Cref{ap:HH-Wasserstein} for a detailed statement and its simple proof.
\end{remark}

\begin{remark}
    [Heavy hitter distance versus L\'evy--Prokhorov distance]
    \label{rem:HH-LP}
The reader may have noticed the similarity of this definition with the L\'evy--Prokhorov metric (though in the L\'evy--Prokhorov metric, the role that balls play in the heavy hitter distance is played by arbitrary measurable sets).  
This naturally suggests the following question: Can we achieve the positive results that are presented below for $\smash{\dDe^{(\ed)}}$ for the L\'evy--Prokhorov metric?
Unfortunately, the lower bounds for Wasserstein distance (\Cref{thm:low-dim-kappa-zero,thm:lb-infinite})
  that we mentioned earlier are easily seen to hold with the exact same parameters for the L\'evy--Prokhorov metric as well. 
\end{remark}

To gain intuition about the new heavy hitter distance notion, let us briefly  discuss the relationship between the  
$\dDe^{(\ed)}(\cdot, \cdot)$ distance notion and the presence of spikes in our distributions. Suppose that $D$ has a spike of amplitude $\ge 2\epsilon$ at $x_0$; in more detail, suppose that  $D=D_1 + D'$ where $D_1 = u_0 \delta_{x_0}$, with $u_0 \ge 2\epsilon$ (think of $\epsilon$ as a small but positive constant), and suppose moreover that this spike is  ``isolated'' in the sense that not only does $D$ have no spikes within (toroidal) distance $\Delta$ of $x_0$, but moreover $D(x) \leq C$ for all $x \ne x_0$ such that $\dtor(x,x_0) \leq \Delta$ (think of $\Delta$ too as a small constant bounded below 1).
This ``isolation" condition implies that while the distribution $D$ may be non-zero in the neighborhood of $x_0$, for $\tau \le \Delta$ we have 
\[
    \int_{ \tau \ge \dtor(x,x_0)>0} D(x) \, dx \le C \cdot \tau^d\,,
\]
which is much less than $\epsilon$ (because of the exponential decay in $d$). 
Consequently, for $\tau \le \Delta$, we have that $D(\Ball(x_0, \tau)) \le u_0 + C \cdot \tau^d$. It follows that if we can do recovery in the heavy hitter distance, i.e., 
if we can obtain $D_2$ such that $\smash{\dDe^{(\ed)}(D, D_2) \le \epsilon}$ for some $\ed \le \Delta/2$,  it would imply that 
\begin{align*}
    D_2\big(\Ball(x_0, \ed)\big) &\le D\big(\Ball(x_0, 2\ed)\big) + \eps \lesssim u_0 + \eps,\quad\text{and}\\[0.6ex]
    D_2\big(\Ball(x_0, \ed)\big) &\ge D\big(\Ball(x_0,0)\big) - \epsilon = u_0 - \eps \geq \epsilon
.\end{align*} Thus, given $D_2$, we can achieve the following goals: 
\begin{itemize}
    \item We can locate the position $x_0$ of the spike in $D$ up to an additive (toroidal) distance of $\ed$;\vspace{-0.15cm}
    \item We can compute the amplitude $u_0$ of the spike in $D$ up to an additive error of roughly $\epsilon$. 
\end{itemize}
Note that these guarantees are meaningful only when the amplitude of the spike $u_0$ exceeds $\eps$ in  absolute value. In the language of data streaming~\cite{woodruff2016new}, we may think of the spike at $x_0$ as being an ``$\epsilon$-heavy hitter.'' 

Finally, we note that recovery under the heavy hitter error metric can be meaningful even without a true ``spike'' ($\delta$-function) at $x_0$. In particular, instead of having $D_1 = u_0 \delta_{x_0}$, if $D_1$ satisfies 
\[
    \int_{\dtor(x,x_0) \le \ed/2} D_1(x)\, dx = u_0\,,
\]
that is, $D_1$ has mass $u_0$ in a ball of radius $\ed/2$ around $x_0$ (an ``approximate spike near $x_0$ at the scale of $\ed/2$''), then a heavy hitter-reconstruction  guarantee stays meaningful.\medskip

\noindent\textbf{Our results for heavy hitter recovery.}  
The following is our main positive result: 

\begin{theorem}[Informal version of \Cref{thm:De-distance-upper}; heavy hitter distance positive result]
\label{thm:De-upper-informal}
	Let $D$ be an unknown distribution over $\mathbb{T}^d$, and fix 
    an error parameter $\eps>0$ and a scale parameter $\ed>0$. Given noisy Fourier coefficients $\{\wt{D}(\ell)\}_{\Vert \ell \Vert_1 \le T}$ of $D$ with 
\[
    T = O\left(\sqrt{d}\cdot \log(1/\epsilon) \cdot \ed^{-1}\right) \quad\text{and}\quad \kappa = \exp \left( - \sqrt{d} \log d \cdot \log(1/\epsilon) \cdot \ed^{-1}\right)
\]
we can recover a distribution  $D'$ such that $\dDe^{(\ed)}(D, D') \le \epsilon$.  
\end{theorem}

Note that for constant $\eps$ and $\ed$, the noise level $\kappa$ is (inverse) \emph{sub-exponential} in the dimension $d$, scaling  as $\smash{2^{-\wt{O}(\sqrt{d})}}$.  
Moreover, the bandlimit is crucially measured in terms of the $\ell_1$ norm (rather than the $\ell_\infty$ norm, as in our Wasserstein results) of the coefficient vector $\ell$. 
This is a significant improvement, since for constant $\eps$ and $\ed$, only $\smash{d^{O(T)} = 2^{\widetilde{O}(\sqrt{d})}}$ Fourier coefficients are required. 
Thus, both the number of Fourier measurements and the required precision are sub-exponential in the ambient dimension $d$. 

It is natural to ask whether \Cref{thm:De-upper-informal} can be quantitatively improved. 
Our main negative result below shows that the 
dependencies on both the bandlimit $T$ and noise level $\kappa$ in \Cref{thm:De-upper-informal} are essentially optimal: 

\begin{restatable} 
[Heavy hitter distance lower bound]
{theorem}
{Dedistancelower}
\label{thm:De-distance-lower}
There is an absolute constant $c>0$ such that the following holds.
Fix $\ed=0.49$ and let $\eps: 2^{-d/3} < \eps < 1/170$. 
There are two probability distributions $D_1$ and $D_2$ over $\mathbb{T}^d$ with the following properties:
\begin{itemize}
\item [(a)] For every $\ell  \in \Z^d$ with $\|\ell\|_1 \leq c\sqrt{{\frac d {\log(1/\eps)}}}$,  we have 
$|\widehat{D_1}(\ell) - \widehat{D_2}(\ell)| \leq  2^{-\Omega\left(\sqrt{d \log(1/\eps)}\right)}$; and\vspace{-0.15cm}

\item [(b)] $\dDe^{(\ed)}(D_1,D_2) > \eps$.
\end{itemize}
\end{restatable}


\subsection{Techniques}

In this subsection we  sketch the main ideas behind our heavy hitter recovery results. 
The Wasserstein results follow more standard arguments and are given in \Cref{sec:appendix-wasserstein}.

\subsubsection{The positive result} 
\label{subsubsec:HH-tech-intro}

Recalling \Cref{thm:De-upper-informal}, we wish to show that if two distributions $D_1, D_2$ differ significantly in heavy hitter distance (that is, $\smash{\dDe^{(\ed)}(D_1, D_2) > \epsilon}$) then there exist an $\ell$ with $\|\ell\|_1 \le T$ such that
\[
    |\widehat{D}_1(\ell) - \widehat{D}_2(\ell)| > \kappa.
\]
Establishing this implication immediately yields recovery: if the noisy $\smash{\{\wt{D}(\ell)\}_{\|\ell\|_1 \le T}}$ are known up to additive error $\kappa$, then one can recover a distribution $D'$ with the desired guarantee.

To prove this, we construct a carefully chosen \emph{mollifier} $p$ that localizes mass in space while remaining bandlimited in the Fourier domain.   
Intuitively, a mollifier is a distribution tightly concentrated around a point (in this case, the origin $0^d$), and we quantify its localization by
\[
    \eta:= \Ex_{\bx \sim p}\big[\dtor(\bx, 0^d)\big]\,.
\]
We refer the reader to~\cite{Rudin:FA:91} for a detailed account of mollifiers and their properties. 
The key challenge is to balance spatial localization (captured by $\eta$) with bandlimit (the smallest $T$ for which $\wh{p}(\ell) = 0$ whenever $\|\ell\|_\infty > T$). 
For our purposes, we construct a mollifier $p$ with the following qualitative properties: 
\begin{itemize}
    \item $p (x) \in [0,1]$ for all  $x \in \mathbb{T}^d$;
    \item $ p(0^d)=1$ and $p(x)$ remains close to $1$ when $\dtor(x,0^d)$ is small; and 
    \item $p(x)$ is close to $0$ when $\dtor(x,0^d)$ is large. 
\end{itemize}
Such a function acts as a smooth indicator of a small neighborhood, allowing us to detect regions where the distribution has significant mass. 

A key point is that heavy hitter recovery only requires a weak form of localization, especially in contrast to Wasserstein distance recovery. 
In particular, even after normalization, if $\bx \sim p$ it may be the case that $\Ex[\dtor(\bx,0^d)]$ is relatively large; hence such a $p$ would not suffice for Wasserstein distance recovery. 
Nevertheless, this weaker form of localization is sufficient for heavy hitter recovery, as it still enables us to reliably identify regions where the distribution places significant mass.  

In \Cref{lem:bump}, we construct such a mollifier $p(\cdot)$ using tools from approximation theory and complexity-theoretic constructions of low-degree polynomials that approximate symmetric Boolean functions over $\zo^d$. 
The key feature of our construction is that $p$ is a trigonometric polynomial whose \emph{total} degree scales only as $\sqrt{d}$. 
This is in sharp contrast to mollifiers that achieve stronger forms of localization (such as those employed in our Wasserstein recovery result, cf.~\Cref{thm:multi-dim-Wass-ub}), for which the total degree must scale at least linearly in $d$. 
Since our measurement model restricts access to Fourier coefficients $\ell$ with $\|\ell\|_1 \le T$, this improved degree bound directly translates into improved bandlimit and noise parameters.

Our construction of the mollifier $p(\cdot)$ proceeds in two stages:
\begin{flushleft} \begin{enumerate}
     \item In stage one, we construct  a simple constant degree trigonometric polynomial which is $1$ at $0^d$ but is at most $\epsilon^2/d$ once $\dtor(x, 0^d) \ge \epsilon$.  
\item In stage two, we compose this with a suitable univariate polynomial $q$ such that \begin{itemize}
    \item[(a)] $q(x) \in [0,1]$ for $x \in [0,1]$; 
    \item[(b)] $q(0)=1$ and $q(x)$ is close to 1 when $x>0$ is close to 0; and 
    \item[(c)] $q(x)$ is close to 0 when $x \in [0,1]$ is large.  
    \end{itemize}
    Crucially, this can be achieved with a polynomial $q$ of degree $\approx \sqrt{d}$.
 \end{enumerate}\end{flushleft}
The key step is the second stage. 
As discussed earlier, the square-root savings in degree arise from results in polynomial approximation theory\hspace{0.1cm}---\hspace{0.1cm}specifically, Paturi’s theorem~\cite{Paturi92} on low-degree polynomial approximations of symmetric Boolean functions. 
This improvement in degree is precisely what allows our algorithm in \Cref{thm:De-upper-informal} to tolerate noise parameters as large as $\kappa = \exp(-\sqrt{d}\log d)$ while requiring Fourier measurements of total degree only $O(\sqrt{d})$ (for fixed $\eps, \ed > 0$). 
Finally, we note that although extremal polynomials with optimal dependence on degree have been widely used in theoretical computer science and machine learning theory~\cite{klivans2001learning,HMPW08,Sherstov:algorithmic}, we are not aware of prior work in the signal processing or 
distribution recovery 
context using them. 

\subsubsection{The negative result}  

The heavy lifting in the proof of \Cref{thm:De-distance-lower} occurs in the proof of an intermediate result, \Cref{thm:cube-lb}, where we construct two distributions $P_1,P_2$ over the hypercube $\bits^d$ with two special properties. To explain these properties, we recall that any function $g: \bits^d \rightarrow \mathbb{R}$ admits a ``Fourier expansion" (sometimes called the Walsh-Fourier expansion), meaning that it can be expressed as 
 \[
g(x) = \sum_{S \subseteq [d]} \widehat{g}(S) \chi_S(x), \quad \textrm{where} \quad \widehat{g}(S) := \Ex_{\bx \sim \bits^d}[g(\bx) \chi_S(\bx)]
\quad \text{and} \quad
\chi_S(x) = \prod_{i \in S} x_i.
 \]
 We refer the reader to \cite{odonnell2014analysis} for background on this rich theory and its many applications in theoretical computer science and mathematics. 
 Turning back to the distributions $P_1$ and $P_2$ (which we view as non-negative functions over $\bits^d$ with $\sum_x P_i(x)=1$), we show that they have the following special properties: 
\begin{flushleft}\begin{enumerate}
    \item  $P_1$ and $P_2$ differ significantly in the amount of mass they put on the point $1^d$; more precisely, $$\big|P_1(1^d) - P_2(1^d)\big| \ge 2\epsilon.$$ 
    \item When we view $P_1$ and $P_2$ as functions over $\bits^d$, their low-degree Fourier coefficients are very close to each other; more precisely, for all $|S| \le c \sqrt{d/{\log(1/\epsilon)}}$, we have that $$\big|\widehat{P_1}(S)- \widehat{P_2}(S)\big| \le 2^{-d} \cdot 2^{\Omega(\sqrt{d \log(1/\epsilon)})}.$$ 
\end{enumerate}\end{flushleft}
 We then map these distributions $P_1$ and $P_2$ over $\bits^d$ to distributions $D_1$ and $D_2$ over $\mathbb{T}^d$  via a simple embedding
that achieves the desiderata of \Cref{thm:De-distance-lower}. 

As the reader can see, the crux of this argument is the construction of the distributions $P_1$ and $P_2$. The existence of such a pair of distributions is argued using constructions of an extremal polynomial~\cite{BEK:97}  which, roughly speaking,  has bounded coefficients but also has a zero of very high multiplicity at $1$. (Interestingly, the arguments in \cite{BEK:97} giving such a real polynomial make essential use of complex analysis.)  We omit further details here, but note that these polynomials and their behavior was crucial in the study of a different  ``inverse statistical'' problem, namely {\em population recovery}, which has been extensively studied over the past decade or so in theoretical computer science~\cite{lovett2017noisy, dvir2012population, moitra2013lossy, polyanskiy2017sample, de2020sharp}. As we elaborate below, the conceptual link to the population recovery literature played a central role in the development of this work.

\subsection{Connections with population recovery}

As mentioned above, the problem of population recovery has been studied in detail in theoretical computer science in recent years. The two most relevant works for us are \cite{polyanskiy2017sample, de2020sharp}, which are independent and concurrent with each other, and which  studied the problem of {\em recovering heavy hitters} from samples on the hypercube that have been contaminated with ``bit-flip noise.'' In this problem, there is an unknown distribution $\mathcal{D}$ over $\{0,1\}^n$. Given a parameter $\rho \in (0,1]$, we define the ``noisy distribution" $\mathcal{D}_\rho$ as follows: to generate a sample from $\mathcal{D}_\rho$, we first sample $\bx \sim \mathcal{D}$ and then flip every bit of $\bx$ independently with probability $(1-\rho)/2$. The goal is to recover (in some sense) the unknown distribution $\mathcal{D}$ from the noisy samples. 

Both \cite{polyanskiy2017sample} and \cite{de2020sharp} showed that for any constants $\rho<1$ and $\epsilon>0$, there is an algorithm which can recover the $\epsilon$-heavy hitters of $\mathcal{D}$ (i.e., the points $x \in \zo^n$ such that $\calD(x) \geq \eps$) with sample complexity $\approx \exp(n^{1/3})$.  Furthermore, for such points, the algorithm also outputs $\mathcal{D}(x)$ up to $\pm \epsilon/2$. In other words, the algorithm outputs the exact location and the approximate mass of each of the heavy hitters.

To connect this to the current work, note that similar to \cite{polyanskiy2017sample, de2020sharp}, we also seek to recover the location of the ``heavy hitters" (i.e.,~the spikes or approximate spikes).  Of course, given that our domain $\mathbb{T}^d$ is continuous, we can only recover the locations of the ``heavy hitters" up to some ``small ball" (this is reflected in our use  of the $\smash{\dDe^{(\ed)}(\cdot, \cdot)}$ reconstruction criterion). A second point of divergence between the current work and \cite{polyanskiy2017sample, de2020sharp} is that while we assume access to noisy \emph{Fourier coefficients}, \cite{polyanskiy2017sample, de2020sharp} assume access to noisy \emph{samples}. These access models are related in that noisy Fourier coefficients can be computed from noisy samples: in particular, for an unknown distribution $\mathcal{D}$, its Fourier coefficient $\mathcal{D}(S)$ can be computed to within additive error $\pm 2^{-d} \cdot \delta$ using  $\approx (1/\delta^2) \cdot (1/\rho)^{|S|}$ noisy samples from $\calD_\rho$. So we can estimate Fourier coefficients from noisy samples (but the sample complexity deteriorates exponentially in the degree of the Fourier coefficient $S$ being estimated). 
However, we cannot go in the other direction and simulate noisy samples given noisy Fourier coefficients; intuitively, this is why \cite{polyanskiy2017sample, de2020sharp} are able to achieve an $\exp(d^{1/3})$ sample complexity while our complexity (for our Fourier coefficient based approach) must scale as $\exp(d^{1/2}).$

We close this discussion by remarking that, while the domain considered in \cite{polyanskiy2017sample, de2020sharp} is different from our setting, as well as the form of access via which the algorithm gets access to the underlying distribution, the \cite{polyanskiy2017sample, de2020sharp} results on sub-exponential time recovery of heavy hitters over the hypercube served as an initial inspiration for our high-dimensional heavy hitter distance results for distribution recovery.

\subsection{Organization}

After reviewing preliminaries in \Cref{sec:prelims}, we give our results for Wasserstein distance reconstruction in \Cref{sec:multi-dim-Wasserstein}, with proofs deferred to \Cref{sec:appendix-wasserstein}. Our main results, upper and lower bounds for heavy hitter reconstruction, are presented in 
\Cref{sec:upperboundheavy,sec:lowerboundheavy}, respectively.



\section{Preliminaries}
\label{sec:prelims}

Throughout this paper, we write $\mathbb{T} := [0,1)$ to denote the torus and $\mathbb{T}^d := [0,1)^d$ to denote the $d$-dimensional torus. 
We write $0^d$ to denote the origin $(0,0,\dots,0) \in \R^d$; we will use similar notation over $\bits^d$ where the meaning will be clear from context.

Throughout this paper, we consider arbitrary probability distributions over $\mathbb{T}^d$,
which may have both continuous and atomic components (i.e., Dirac deltas). 
With a mild (and standard) abuse of notation, we identify such distributions with $L^1$-integrable functions $D : \mathbb{T}^d \to \mathbb{R}_{\ge 0}$ satisfying
\[
    \int_{\mathbb{T}^d} D(x)\,dx = 1,
\]
while allowing atomic delta components. 
(This convention lets us treat continuous and discrete distributions in a unified manner.)  
For measure-theoretic background, we refer the reader to e.g., Chapter~6 of~\cite{rudin1987real} or Chapter~3 of~\cite{folland1999real}.

\subsection{Fourier analysis over $\mathbb{T}^d$}\label{sec:fourier-torus}

Recall that over $\mathbb{T}^d$, for $\ell=(\ell_1,\dots,\ell_d) \in \Z^d$,  the $\ell^\text{th}$ Fourier coefficient of a distribution $D$ is 
\[
\widehat{D}(\ell) = \int_{x \in \mathbb{T}^d} D(x) \cdot e^{2 \pi i (\ell \cdot x)} dx.
\]
For a ``point mass'' at $r \in \mathbb{T}^d$ (corresponding to the Dirac delta function $D(x) = \delta_r(x) =  \delta(x-r)$), the Fourier transform is given by
\[
\widehat{D}(\ell) = \int_{x \in \mathbb{T}^d} \delta(x-r) \cdot e^{2 \pi i (\ell \cdot x)} dx = e^{2 \pi i (\ell \cdot r)}.
\]
We recall \emph{Plancherel's theorem}, which says that if $D \in L^2$, then $\int_{x \in \mathbb{T}^d} |f(x)|^2 dx = 
\sum_{\ell \in \Z^d} |\widehat{f}(\ell)|^2$. 
We further recall that given distributions $D_1$ and $D_2,$, their \emph{convolution}, written $D_1 * D_2$, is
\[
(D_1 \ast D_2) (x) = \int_{y \in \mathbb{T}^d} D_1(y) \cdot D_2(x-y) \, dy,
\]
and given that $D_1,D_2$ are $L^1$-integrable, we have $\widehat{D_1 * D_2}(\ell)=\widehat{D_1}(\ell) \cdot \widehat{D_2}(\ell)$ for all $\ell \in \Z^d$.

\subsection{Distance and Wasserstein distance over the torus}
\label{subsec:prelims-wass}

Recall that the ``toroidal distance'' $\dtor$ 
between two points $a_1, b_1$ in the one-dimensional torus $\mathbb{T}$ (also known as the Lee metric) is 
\[
\dtor(b_1,a_1) = \dtor(a_1,b_1) := \min\big\{b_1-a_1, a_1+(1-b_1)\big\} \quad \text{where~}0 \leq a_1 < b_1,
\]
and that the Euclidean toroidal distance between two points $a,b$ in the $d$-dimensional torus $\mathbb{T}^d$ is
\[
\dtor(a,b) = \sqrt{\dtor(a_1,b_1)^2 + \cdots + \dtor(a_d,b_d)^2}.
\]

We recall that if $f,g: \mathbb{T}^d \to \R_{\geq 0}$ are non-negative valued functions with $\int_{x \in \mathbb{T}^d} f(x) dx=\int_{x \in \mathbb{T}^d} g(x) dx=1$, corresponding to the densities of two probability distributions, then the \emph{Wasserstein distance} $\dW(f,g)$ between $f$ and $g$ is the minimum ``transportation cost'' of turning $f$ into $g$, i.e.,~the minimum, over all couplings $(\bX,\bY)$ of random variables $\bX \sim f$, $\bY \sim g$, of $\E[\dtor(\bX,\bY)]$.
As is well known and was mentioned earlier, by Kantorovich-Rubenstein duality an alternative definition is that

\begin{equation} \label{eq:KR}
\dW(f,g) := \sup_{h: \mathbb{T}^d \rightarrow \mathbb{R}} \int_{x} h(x) \cdot (f(x) - g(x)) dx\,,
\end{equation}
where $h$ ranges over all $1$-Lipschitz (vis-a-vis toroidal distance) functions over $\mathbb{T}^d$ satisfying $\|h\|_\infty \leq \sqrt{d}$.   (Note that if $f$ and $g$ are densities of probability distributions, then the condition $\|h\|_\infty \leq \sqrt{d}$ is immaterial since \Cref{eq:KR} is unchanged when $h(x)$ is replaced by $h(x)+c$ for any $c \in \R$.)

We will use the following fact, which gives a bound on the Fourier coefficients of 1-Lipschitz functions (see,  e.g.,~\cite{HSSV21colt} for the simple argument establishing this):

\begin{fact} \label{fact:Lipschitz-Fourier}
Let $f$ be any 1-Lipschitz function over $\mathbb{T}^d$ and let $0^d \neq \ell \in \Z^d$. Then it holds that 
\[
	|\widehat{f}(\ell)| \le \frac{1}{  \|\ell\|_2}\,.
\]
\end{fact}


\section{Wasserstein distance recovery}
\label{sec:multi-dim-Wasserstein}

In this section, we state our upper and lower bounds for Wasserstein distance reconstruction in the $d$-dimensional setting. 
All proofs are deferred to \Cref{sec:appendix-wasserstein}.

\subsection{Upper bound}

We prove that reconstruction within Wasserstein distance $\eps$ is possible from noisy Fourier coefficients $\{\wt{f}(\ell)\}_{\|\ell\|_\infty \le T}$ with bandlimit $T = O(\sqrt{d}/\eps)$, provided each coefficient is corrupted by additive noise of magnitude at most $\kappa = O(\eps^d / d^{d/2})$. 
Moreover, for fixed dimension $d$, the reconstruction can be carried out in time polynomial in the input size via linear programming; see \Cref{sec:wasserstein-algo}. 

\begin{theorem}~\label{thm:multi-dim-Wass-ub}
    Let $D_1$ and $D_2$ be two distributions over the the $d$-dimensional torus $\mathbb{T}^d$. For any $0 < \eps < 1/2$, and any $d \geq 1$, 
    \[
    \text{let~}T = \left\lceil 6\sqrt{d}/\eps \right\rceil,
    \quad \text{and let~} 
    \kappa = \begin{cases}
        0.001 \eps / \log(1/\eps) & \text{~when~}d = 1\\
        \left(0.01 \eps / \sqrt{d}\right)^d & \text{~when~}d \ge 2
        \end{cases}.
    \]
        If $D_1,D_2$ satisfy $| \widehat{D_1}(\ell) - \widehat{D_2}(\ell) | \le \kappa$ for all integer vectors $\ell = (\ell_1, \dots, \ell_d) \in \mathbb{Z}^d$ with $\left\| \ell \right\|_{\infty} \le T$, then $\dW(D_1, D_2) \le \eps$.
\end{theorem}

It is instructive to see how the Wasserstein error guarantee in \Cref{thm:multi-dim-Wass-ub} circumvents the need for a separation condition.
For example, returning to the distributions from \Cref{eq:ithaca}, note that the two signals $D_1$ and $D_2$ there are $\Delta/2$-close in Wasserstein distance.   
Intuitively, the Wasserstein metric blurs the distinction between nearby points, allowing for meaningful approximate recovery even when spikes are arbitrarily close.

Of course, if the support points are well separated, then accurate recovery in Wasserstein distance automatically yields accurate estimates of the true support locations $\{t_j\}$ and  amplitudes $\{u_j\}$.
In this sense, the Wasserstein criterion provides a \emph{soft interpolation} between rigid separation-based guarantees and meaningful recovery in the absence of any separation assumption. 

\begin{remark} 
\label{rem:MMRS}
As mentioned earlier, \Cref{thm:multi-dim-Wass-ub} is closely related to Theorem~32 of~\cite{MMRS25}. 
In particular, the results of~\cite{MMRS25} imply Wasserstein recovery with bandlimit $T = O(d/\epsilon)$ and noise level $\kappa = O(1/T^d)$ (their formulation allows for non-uniform noise across measurements; fixing a uniform noise level yields this bound). 
Their approach is formulated in terms of approximate Chebyshev moment measurements, whereas ours uses Fourier coefficients. 
However, these perspectives are closely connected via classical approximation theory and can be viewed as morally equivalent. 

Our result slightly improves the dependence on the dimension by achieving bandlimit $T = O(\sqrt{d}/\epsilon)$ with essentially the same requirements on the noise parameter $\kappa$. 
That said, in the applications considered in~\cite{MMRS25}, the dominant constraint is the noise level, which scales as $\exp(-\Theta(d \log d))$, rather than the bandlimit. 
From this perspective, it is plausible that the $\sqrt{d}$ improvement in $T$ could be obtained in the \cite{MMRS25} setting via a more careful analysis of their approach.
\end{remark}

\subsection{Lower bounds}

\Cref{thm:low-dim-kappa-zero} shows that even given the exact Fourier coefficients (i.e., in the noiseless setting when $\kappa = 0$), if the bandlimit $T$ on individual degree satisfies $T \lesssim \sqrt{d}/(2 \epsilon)$, then  $\eps$-accurate recovery in Wasserstein distance is impossible. It is achieved by considering two distributions which are uniform on grids of width $1/T$ where the two grids are $1/T$-shifts of each other (for $T \approx \sqrt{d}/2\epsilon$). As the distributions are uniform on grids, calculating their Fourier coefficients and showing that they are identical up to bandlimit $T$ is a straightforward exercise. 

\begin{restatable}{theorem}{WassLBKappaZero}
\label{thm:low-dim-kappa-zero}
    For any $0 < \eps \le 1/2$, there are two distributions $D_1$, $D_2$ over the $d$-dimensional torus $\mathbb{T}^d$, such that $\widehat{D_1}(\ell) = \widehat{D_2}(\ell)$ for any $\ell \in \mathbb{Z}^d$ with $\|\ell\|_\infty \le \sqrt{d}/(2\eps) - 1$, while $\dW(D_1, D_2) \ge \eps$.
\end{restatable}

\Cref{thm:lb-infinite} shows that even given all (infinitely many) Fourier coefficients to error $\kappa = \epsilon^{0.249 d}$, the minimum error in Wasserstein recovery must be at least $\epsilon$. More precisely, it establishes the existence of two distributions $D_1$ and $D_2$ such that $|\widehat{D_1}(\ell) - \widehat{D_2}(\ell)| \leq \kappa$ for all $z \in \Z^d$, yet $\dW(D_1,D_2)> \eps.$ 
The lower bound construction uses two distributions each of which is the convolution of a suitable kernel with a discrete distribution that is uniform on a randomly chosen discrete point set of size $\approx (1/\epsilon)^d$.
The proof of correctness uses the probabilistic method in tandem with the fact that convolution with this specific kernel annihilates higher-order Fourier coefficients. 

\begin{restatable}{theorem}{WassLBInfiniteT}
\label{thm:lb-infinite}
    There exists a universal constant $\eps_0 > 0$, such that for any $0 < \eps \le \eps_0$, there are two distributions $D_1$, $D_2$ over the $d$-dimensional torus $\mathbb{T}^d$, such that $|\widehat{D_1}(\ell) - \widehat{D_2}(\ell)| \le \eps^{0.249d}$ for any $\ell \in \mathbb{Z}^d$, while $\dW(D_1, D_2) \ge \eps$.
\end{restatable}

Finally, a simple counterexample (\Cref{obs:1dlowerinfiniteT}) demonstrates that if the noise level $\kappa$ is allowed to be as large as $O(\eps)$, then even access to \emph{all} Fourier coefficients (i.e., infinite bandlimit) does not permit $\eps$-accurate Wasserstein recovery. 



\section{A $d$-dimensional upper bound for heavy hitter reconstruction}\label{sec:upperboundheavy}

We now turn to the main positive result of this paper, which is an upper bound for heavy hitter reconstruction of distributions.
Our goal in this section is to prove the following:

\begin{theorem} 
\label{thm:De-distance-upper}
	Given $0 < \eps \leq 1/2, 0 < \ed < 1$, let $D_1,D_2$ be two distributions over $[0,1)^d$ such that 
	\begin{equation} \label{eq:DDU}
		\abs{\widehat{D_1}(\ell) - \widehat{D_2}(\ell)} \leq \kappa := 
		\frac{\eps}{9}\pbra{\frac{1}{d}}^{O(T)} 
		\text{~for all $\ell \in \Z^d$ with~}\|\ell\|_1 \leq T\,,
	\end{equation}
	for some \smash{$T = O\pbra{{\frac {\sqrt{d} \log(1/\eps)}{\ed}}}$}. 
	Then $\dDe^{(\ed)}(D_1,D_2) \leq \eps.$
\end{theorem}

\medskip

Note that in \Cref{thm:De-distance-upper} the degree bound is in the sense of \emph{total degree} rather than \emph{individual degree} as was the case for our Wasserstein results. 
As mentioned earlier, this is important because it means that the total number of Fourier coefficients required for reconstruction scales (as a function of $d$) as $2^{\wt{O}(\sqrt{d})}$ rather than as $\exp(d)$.


\subsection{A low-degree trigonometric polynomial ``bump function''} \label{sec:trigbump}

An essential ingredient of our approach is a (relatively) low-degree trigonometric polynomial which serves as a mollifier or ``bump function'' over the torus $[0,1)^d$; the main result of this subsection, \Cref{lem:bump}, establishes the existence of such a trigonometric polynomial.

\subsubsection{A useful univariate polynomial}

To build the desired multivariate trigonometric polynomial bump function, we begin with a  construction of a univariate polynomial with the following behavior over the interval $[0,1]$:  it takes values in $[0,1]$ on inputs in $[0,1]$, takes values close to 1 on inputs that are within a certain prescribed distance of 0, and takes values close to 0 on inputs that are ``even a little bit beyond'' that prescribed distance.  The following result, giving such a univariate polynomial, is the main result of this subsubsection:

\begin{lemma}  \label{lem:univariate-new}
Let $0<\eps,\ed<1$, let $0<b<A$ where $A,b=\Theta(\ed^2)$, and let $d$ be an asymptotic parameter.
There is a real polynomial $q(x)$ of degree
\begin{equation} \label{eq:maybeblah}
\Psi(\eps,A,b,d)
= O\pbra{
{\frac {\sqrt{d} \cdot \log(1/\eps)} {\ed}}
}
\end{equation}
with the following properties:
\begin{itemize}
\item [(1)] For $x \in [0,1]$ we have $q(x) \in [0,1]$;
\item [(2)] $q(x) \in [1-\eps/2,1]$ for $x \in [0,A/d]$;
\item [(3)] $q(x) \in [0,\eps/8]$ for $x \in [(A+b)/d,1]$.
\end{itemize}
\end{lemma}

We remark in passing that while properties (1)--(3) may remind the reader of the Chebyshev polynomial, a shifted and scaled version of the Chebyshev polynomial does not by itself suffice to yield \Cref{lem:univariate-new}. Roughly speaking, this is because in \Cref{lem:univariate-new} the width $A/d$ of the interval over which the polynomial takes values very close to 1 (see item (2)) is within a constant factor of the width $b/d$ of the interval over which the polynomial changes from being very close to 1 to being very close to 0 (see items~(2) and (3)); this is not how the Chebyshev polynomial behaves.

We remark that in \cite{Sherstov:algorithmic} Sherstov gives sharp upper and lower bounds for a related problem, by constructing optimal (in terms of degree) polynomials which appproximately achieve a prescribed sequence of 0/1 values at the discrete points $x=0,{\frac 1 d}, {\frac 2 d}, \dots, {\frac {d-1} d}, 1$ of the real line. In our setting, though, we require high-accuracy approximators over the \emph{continuous} intervals $[0,A/d]$ and $[(A+b)/d,1]$, not just at discrete values.
 
To prove \Cref{lem:univariate-new} we will use the following result, which is implicit in the proof of Theorem~3 (the main upper bound) of \cite{Paturi92}:

\begin{lemma}
\label{lem:Paturi}
Fix any $\ell \in [0,m/2)$ and let $g: [0,m] \to [0,1]$ be the piecewise linear continuous function defined by
\[
g(x) = 
\begin{cases}
1 & \text{if~}x \in [0,\ell]\\
1 - (x-\ell)& \text{if~} x \in (\ell,\ell+1)\\
0 & \text{if~} x \in [\ell+1,m]
\end{cases}
.
\]
Then there is a polynomial $p$ of degree at most $O(\sqrt{m \ell})$ such that $|p(x)-g(x)| \leq 1/3$ for all $x \in [0,m]$. (So taking $r(x) = {\frac 3 5}(p(x)+{\frac 1 3})$, we get that $r(x) \in [0,1]$ for all $x \in [0,m]$; $r(x) \in [{\frac 3 5},1]$ for all $x \in [0,\ell]$; and $r(x) \in [0,{\frac 2 5}]$ for all $x \in [\ell+1,m].$)
\end{lemma}

We will combine \Cref{lem:Paturi} with a standard construction of ``amplifying polynomials'' (see e.g.~Claim~3.8 of \cite{DGJ+:10}).
In more detail, let 
\begin{equation} \label{eq:amplifier}
a_k(t):=\Pr[\text{$k$ tosses of a coin with heads probability $t$ yield at least $k/2$ heads}].
\end{equation}
It is immediate that $a_k$ is a degree-$k$ polynomial, since
\[
a_k(t) = \sum_{j=\lceil k/2 \rceil}^k {k \choose j} t^j (1-t)^{k-j}.
\]
Moreover, a standard additive Hoeffding bound gives that for $k=O(\log(1/\tau))$,
\[
a_k(t) \begin{cases}
\geq 1 - \tau & \text{~for~}t \in [3/5,1]\\
\leq \tau & \text{~for~}t \in [0,2/5]
\end{cases}.
\]
Putting the pieces together, we see that given $\eps,\ed,A,b$ as in the statement of \Cref{lem:univariate-new}, taking \Cref{lem:Paturi}'s $m$ parameter to be $d/\Theta(\ed^2)$, its $\ell$ parameter to be $\Theta(1)$, the $k$ parameter of \Cref{eq:amplifier} to be $O(\log(1/\tau))$ and $\tau$ to be $\eps/8$, 
we get that the polynomial
\[
q(x) := a_k(r(x)),
\quad
\text{for the polynomial $r$ given in \Cref{lem:Paturi}},
\]
satisfies properties (1), (2) and (3) of \Cref{lem:univariate-new}, and \Cref{lem:univariate-new} is proved. \qed

\subsubsection{The multivariate trigonometric polynomial ``bump function''}

With \Cref{lem:univariate-new} in hand, we  build the desired multivariate ``bump function'' $p(x_1,\dots,x_d)$ over the torus $[0,1)^d$ straightforwardly as follows: 
\begin{equation} \label{eq:bump}
p(x_1,\dots,x_d) := q\pbra{
{\frac {\sin^2(\pi x_1) + \cdots + \sin^2(\pi x_d)} d}
}.
\end{equation}

\begin{lemma} \label{lem:bump}
The polynomial $p(x_1,\dots,x_d)$ is a trigonometric polynomial of total degree at most $2\Psi(\eps,A,b,d)$ (where $\Psi$ is as in~\Cref{lem:univariate-new}) with the following properties: for $\eps,\ed,A,b$ as in \Cref{lem:univariate-new},

\begin{itemize}

\item [(i)] For $x \in [0,1)^d$ we have $p(x) \in [0,1]$;

\item [(ii)] If $\dtor(x,0^d) \leq {\frac {\sqrt{A}} \pi},$ then $p(x) \in [1-\eps/2,1]$;

\item [(iii)] If $x \in [0,1)^d$ has $\dtor(x,0^d) \geq {\frac {\sqrt{A+b}} 2},$ then $p(x) \in [0,\eps/8].$
\end{itemize}

\end{lemma}

\begin{proof}
The degree bound is immediate from the fact that the polynomial $q$ from \Cref{lem:univariate-new} has degree $\Psi(\eps,A,b,d)$.
Item (i) follows immediately from item~(1) of \Cref{lem:univariate-new}, since the argument ${\frac {\sin^2(\pi x_1) + \cdots + \sin^2(\pi x_d)} d}$ to $q$ in \Cref{eq:bump} is always between 0 and 1.

For (ii), fix any $x$ such that  $\dtor(x,0^d) \leq {\frac {\sqrt{A}} \pi}$, i.e.
\[
\dtor(x_1,0)^2 + \cdots + \dtor(x_d,0)^2 \leq {\frac A {\pi^2}}.
\]
For each $x_i \in [0,1)$ we have that $\sin^2(\pi x_i) = \sin^2(\pi \dtor(x_i,0)) \leq \pi^2 \dtor(x_i,0)^2,$ so
\[
{\frac {\sin^2(\pi x_1)+\cdots+\sin^2(\pi x_d)}{d}} \leq 
\pi^2 \cdot {\frac {\dtor(x_1,0)^2 + \cdots + \dtor(x_d,0)^2}{d}} \leq {\frac {A}{d}}, \quad \text{and hence}
\]
\[
p(x) = q \pbra{
{\frac {\sin^2(\pi x_1)+\cdots+\sin^2(\pi x_d)}{d}}
} \in [1-\eps/2,1]
\]
by part (2) of \Cref{lem:univariate-new}, as desired.

For item~(iii) we use the following claim:
\begin{claim} \label{claim:dist}
Let $x$ be a point in the $n$-dimensional torus $[0,1)^d$ such that $\dtor(x,0^d) \geq \gamma$.  
Then 
\[
{\frac {4 \gamma^2} d} \leq {\frac {\sin^2(\pi x_1) + \cdots + \sin^2(\pi x_d)}{d}} \leq 1.
 \]
\end{claim}
\begin{proofof}{\Cref{claim:dist}}
The upper bound is immediate.  For the lower bound, fix any $x$ such that $\dtor(x,0^d) \geq \gamma$, i.e.
\[
\dtor(x_1,0)^2 + \cdots + \dtor(x_d,0)^2 \geq \gamma^2.
\]
For any $x_i \in [0,1)$ we have that $\sin^2(\pi x_i) \geq 4 \dtor(x_i,0)^2$, so
\[
{\frac {\sin^2(\pi x_1) + \cdots + \sin^2(\pi x_d)} 4} \geq
\dtor(x_1,0)^2 + \cdots + \dtor(x_d,0)^2 \geq \gamma^2
\]
which directly gives the lower bound of \Cref{claim:dist}.
\end{proofof}
Taking $\gamma = {\frac {\sqrt{A+b}} 2}$, \Cref{claim:dist} gives that the argument ${\frac {\sin^2(\pi x_1) + \cdots + \sin^2(\pi x_d)}{d}}$ to $p$ in \Cref{eq:bump} is at least ${\frac {A+b} d}$, and item (iii) follows from item (3) of \Cref{lem:univariate-new}.
This concludes the proof of \Cref{lem:bump}.
\end{proof}

\subsection{Proof of \Cref{thm:De-distance-upper}}

We prove the theorem via its contrapositive. 
Fix two distributions $D_1,D_2$ over the torus $[0,1)^d$ and suppose that $\dDe^{(\ed)}(D_1,D_2) \geq \eps.$
By \Cref{def:De-distance} (interchanging the roles of $D_1$ and $D_2$ if necessary), this means that there exist some $z \in [0,1)^d$ and some 
$0 \leq \tau \leq \ed$ such that 
\begin{equation}
D_1(\Ball(z, \tau)) > D_2(\Ball(z,\tau+\ed)) + \eps. \label{eq:strawberry}
\end{equation}
We may suppose that $z=0^d$ (this is without loss of generality by shifting the argument to $p$ by $z$ in what follows).
We fix the values of $A$ and $b$ as a function of $\ed$ depending on the value of $\tau$: 
\begin{itemize}
    \item If $0.6\ed \le \tau\le \ed$, then we set $A$ and $b$ as follows:
    \[
        A:= (\pi \ed)^2
        \quad\text{and}\quad
        b:=\big((3.2)^2-\pi^2\big)\ed^2\,.
    \]
    \item If $0\le \tau<0.6\ed$, then we set $A$ and $b$ as follows: 
    \[
        A:=(0.6\pi\ed)^2
        \quad\text{and}\quad
        b:=\big(4-(0.6\pi)^2\big)\ed^2\,.
    \]
\end{itemize}
It is readily checked that in both cases, we have $0<b<A$ where $A,b=\Theta(\ed^2)$, and moreover
\begin{equation} \label{eq:stitch}
    \tau \leq \frac{\sqrt{A}}{\pi} 
    \qquad\text{and}\qquad
    \tau + \ed \geq \frac{\sqrt{A + b}}{2}\,.
\end{equation}

In order to establish \Cref{thm:De-distance-upper}, we will show that $|\widehat{D_1}(\ell) - \widehat{D_2}(\ell)| > \kappa$ for some $\ell=(\ell_1,\dots,\ell_d)$ with 
$\|\ell\|_1 \leq 2T$, where 
\[
    \kappa = \frac{\eps}{3}\pbra{\frac{1}{d}}^{O(T)}
    \qquad\text{and}\qquad 
    T = \Psi(\eps, A,b,d)\,, 
\]
for $\Psi(\eps,A,b,d)$ from \Cref{lem:univariate-new}.  

Suppose to the contrary that $|\widehat{D_1}(\ell) - \widehat{D_2}(\ell)| \leq \kappa$  for every $\ell \in \Z^d$ with $\|\ell\|_1 \leq 2T$. 
Let $p(x_1, \dots, x_d)$ be the polynomial from~\Cref{lem:bump}. 
We then have 
\begin{align}
\Ex_{\bx \sim D_1}\sbra{p(\bx)} &=
\int_{x \in [0,1)^d} D_1(x) p(x) \, dx \nonumber \\
&\geq
\int_{x \in \Ball(0^d,\tau)} D_1(x) p(x) \, dx \tag{since $D_1,p \geq 0$ on $[0,1)^d$}\\
&\geq
(1-\eps/2)\cdot D_1\pbra{\Ball(0^d, \tau)}\,,
\label{eq:o}
\end{align}
where \Cref{eq:o} uses $\tau \leq {\frac {\sqrt{A}} \pi}$ (see~\Cref{eq:stitch}) and item~(ii) of \Cref{lem:bump}.

We further have that 
\begin{align}
    \Ex_{\bx \sim D_2}\sbra{p(\bx)} 
    &=
    \int_{x \in [0,1)^d} D_2(x) p(x) \, dx \nonumber \\
    &=\int_{x \in [0,1)^d} D_2(x) \pbra{\sum_{\|\ell\|_1 \leq 2T} \wh{p}(\ell)e^{-2\pi i \ell x}} \,dx \tag{by the total degree bound on $p$} \\
    &= \sum_{\|\ell\|_1 \leq 2T} \wh{p}(\ell) \pbra{\int_{[0,1)^d} D_2(x) e^{-2\pi i \ell x} \,dx } \tag{Fubini} \\
    &= \sum_{\|\ell\|_1 \leq 2T} \wh{p}(\ell)\wh{D_2}(-\ell) \nonumber \\
    &= \sum_{\|\ell\|_1 \leq 2T} \wh{p}(\ell)\overline{\wh{D_2}(\ell)} = \sum_{\|\ell\|_1 \leq 2T} \overline{\wh{p}(\ell)}{\wh{D_2}(\ell)}\,, \label{eq:melodrama}
\end{align}
where the final line relies on the fact that as $D_2$ is real valued (since it is a probability density), we have 
\[
    \wh{D_2}(-\ell) = \overline{\wh{D_2}(\ell)}\,.
\]
Continuing, we thus have 
\begin{align}
    \Ex_{\bx \sim D_2}\sbra{p(\bx)} 
    &= \sum_{\|\ell\|_1 \leq 2T} \overline{\wh{p}(\ell)}{\wh{D_2}(\ell)} \\
    &= \sum_{\|\ell\|_1 \leq 2T} \overline{\wh{p}(\ell)}{\wh{D_1}(\ell)} + \sum_{\|\ell\|_1 \leq 2T} \overline{\wh{p}(\ell)}\pbra{\wh{D_2}(\ell)-\wh{D_1}(\ell)} \nonumber \\
    &= \Ex_{\bx\sim \calD_1}\sbra{p(\bx)} + \sum_{\|\ell\|_1 \leq 2T} \overline{\wh{p}(\ell)}\pbra{\wh{D_2}(\ell)-\wh{D_1}(\ell)} \label{eq:lorde-sober}\,,
\end{align}
where \Cref{eq:lorde-sober} relies on repeating the steps giving \Cref{eq:melodrama} with $D_1$ instead of $D_2$. 
Next, note that $|\overline{\wh{p}(\ell)}| \leq 1$ for all $\ell$ as $p \in [0,1]$ and additionally we assumed $|\wh{D_1}(\ell) - \wh{D_2}(\ell)| \leq \kappa$ for all $|\ell|_1 \leq 2T$. 
Since there are at most $d^{O(T)}$ vectors $\ell \in \Z^d$ with $\|\ell\|_1 \leq 2T$, we can continue from \Cref{eq:lorde-sober} as follows:
\begin{align} 
    \Ex_{\bx \sim D_2}\sbra{p(\bx)}
    &\geq \Ex_{\bx\sim \calD_1}\sbra{p(\bx)} - \kappa\cdot d^{O(T)} \label{eq:jungle-keep-moving} \\
    &\geq (1-\eps/2)\cdot D_1\pbra{\Ball(0^d, \tau)} - \kappa\cdot d^{O(T)} \tag{using \Cref{eq:o}} \\
    &\geq (1-\eps/2) \cdot \pbra{D_2\pbra{\Ball(0^d, \tau + \ed)}  + \eps}  - \kappa d^{O(T)}  \label{eq:the}\\
    &\geq D_2\pbra{\Ball(0^d, \tau + \ed)}  + \eps/2 - \eps^2/2 -  \kappa d^{O(T)} \label{eq:thethe}\\
        &\geq D_2\pbra{\Ball(0^d, \tau + \ed)}  + \eps/4 -  \kappa d^{O(T)} \,. \label{eq:thethethe}
\end{align} 
where \Cref{eq:the} uses \Cref{eq:strawberry}, \Cref{eq:thethe} uses that $D_2\pbra{\Ball(0^d, \tau + \ed)} \leq 1$, and \Cref{eq:thethethe} uses that $\eps < 1/2$.

On the other hand, we also have that
\begin{align}
    \Ex_{\bx \sim D_2}\sbra{p(\bx)} 
    &= 
    \int_{x \in \Ball(0^d,\tau + \ed)} D_2(x)p(x) \,dx +
    \int_{x \notin \Ball(0^d,\tau + \ed)} D_2(x)p(x) \,dx \nonumber\\
    &\leq D_2(\Ball(0^d,\tau+\ed)) + \eps/8, \label{eq:live}
\end{align}
where \Cref{eq:live} uses item (iii) of \Cref{lem:bump}. (Note that we indeed have that $\tau + \ed \geq {\frac {\sqrt{A+b}} 2}$ as required for our application of item (iii) of \Cref{lem:bump} thanks to \Cref{eq:stitch}.) 
Since
$\kappa = (\eps/9)/d^{O(T)}$, \Cref{eq:thethethe} and \Cref{eq:live} together give a contradiction. This concludes the proof of \Cref{thm:De-distance-upper}.
\qed



\section{A $d$-dimensional lower bound for heavy hitter reconstruction}\label{sec:lowerboundheavy}


Our goal in this section is to prove the following:


\Dedistancelower*

In words, \Cref{thm:De-distance-lower} shows that the upper bound given in \Cref{thm:De-distance-upper} is nearly best possible: there exist two distributions  whose low-degree (almost up to the degree bound of \Cref{thm:De-distance-upper}, off by only a $\log(1/\eps)$ factor) Fourier coefficients ``nearly match each other'' (where the closeness is again quantitatively quite similar to that which \Cref{thm:De-distance-upper} says would be sufficient for small heavy hitter distance), yet the heavy hitter distance between these two distributions is greater than $\eps.$

\subsection{A lower bound over the discrete cube $\bits^d$}\label{sec:lb-discrete-cube}

Most of the work in establishing \Cref{thm:De-distance-lower} is in proving \Cref{thm:cube-lb}, which shows the existence of two specially crafted probability distributions $P_1,P_2$ over the discrete cube $\bits^d$.  
As we will see, the existence of these distributions follows from the existence of certain univariate polynomials with many repeated roots at 1 and a large-magnitude constant term, established by Erd\'{e}lyi \cite{Erdelyi16}.

Before stating the theorem, we introduce some basic notation for distributions over $\bits^d$.  For a distribution $P$ over $\bits^d$, we write $P(x)$ to denote the amount of probability weight that $P$ puts on the point $x \in \bits^d$, i.e. we view $P(x)$ as a non-negative function over $\bits^d$ such that $\sum_{x } P(x)=1$.  For $S \subseteq \{1,\dots,d\}$ we write $\widehat{P}(S)$ to denote the ``$S$-th Fourier coefficient'' of $P$, namely $$\widehat{P}(S) = \Ex_{\bx \sim \bits^d}\big[P(\bx) \chi_S(\bx)\big]$$ where $\chi_S: \bits^d \to \bits$ is $\chi_S(x) = \prod_{i \in S} x_i$.
To get a sense of the right scaling for the Fourier coefficients of a distribution over $\bits^d$, note that for any distribution $P$ over $\bits^d$ we have
\[
\sum_{S \subseteq \{1,\dots,d\}} \widehat{P}(S)^2 = \Ex_{\bx \sim \bits^d}\left[P(\bx)^2\right] \in \left[2^{-2d},2^{-d}\right],
\]
so every Fourier coefficient $\widehat{P}(S)$ satisfies $|\widehat{P}(S)| \leq 2^{-d/2}$.

The following theorem says that there are two distributions $P_1,P_2$ that put significantly different amounts of weight on the point $1^d$ yet have close-to-matching low-degree Fourier coefficients:

\begin{theorem} \label{thm:cube-lb}
There is a suitable absolute constant $c>0$ such that the following holds.
Fix an $\eps:2^{-d/3} < \eps < 1/170$.
There are two distributions $P_1,P_2$ over $\bits^d$ with the following properties:
\begin{itemize}
\item [(1)] For every $S \subset \{1,\dots,d\}$ with $|S| \leq c\sqrt{{\frac d{\log(1/\eps)}}}$, we have 
$\left|\widehat{P_1}(S) - \widehat{P_2}(S)\right| \leq 2^{-d} \cdot  2^{-\Omega\left(\sqrt{d \log(1/\eps)}\right)}.$

\item [(2)] $\big|P_1(1^d) - P_2(1^d)\big| \geq 2\eps$.
\end{itemize}
\end{theorem}

We prove \Cref{thm:cube-lb} in the rest of 
\Cref{sec:lb-discrete-cube}, and use it to prove \Cref{thm:De-distance-lower} in \Cref{subsec:torus}.

\subsubsection{Setup for the proof of \Cref{thm:cube-lb}}

Given $t \geq 0$, let $\Prod_t$ denote the i.i.d.~product distribution over $\bits^d$ defined as follows: for each $i \in \{1,\dots,d\}$, we have $\E_{\bx \sim \Prod_t}[\bx_i] = e^{-t}.$
For any $z \in \bits^d$ that has $\ell$ coordinates that are 1 and $d-\ell$ coordinates that are $-1$, the probability that $\Prod_t$ puts on $z$ is 
\begin{equation} \label{eq:formula-for-Prod}
\Prod_t(z) = 2^{-d}(1+e^{-t})^\ell(1-e^{-t})^{d-\ell}.
\end{equation}
Moreover, for any $S \subseteq \{1,\dots,d\}$, we have that the $S$-th Fourier coefficient of $\Prod_t: \bits^d \to \R$ is
\begin{equation}
\label{eq:ProdhatS}
\widehat{\Prod_t}(S) = 
2^{-d} \sum_{x \in \bits^d} \Prod_t(x) \chi_S(x)
= 2^{-d} \Ex_{\bx \sim \Prod_t}\big[\chi_S(\bx)\big] 
= 2^{-d} \prod_{j \in S} \Ex_{\bx \sim \Prod_t}[\bx_j]
= 2^{-d} e^{-t|S|}.
\end{equation}

Given a vector of non-negative probabilities $\bar{p}=(p_0,p_1,\dots,p_d)$ that sum to 1, and a vector of non-negative values $\bar{t}=(t_0,t_1,\dots,t_d)$, we write $\Mix_{\bar{p},\bar{t}}$ to denote the mixture distribution that puts weight $p_i$ on component distribution $\Prod_{t_i}$ for each $i=0,1,\dots,d$; so $\Mix_{\bar{p},\bar{t}}$ is a distribution over $\bits^d$.
By \Cref{eq:ProdhatS} and the linearity of the Fourier transform, we have that 
\begin{equation}
\label{eq:MixhatS}
\widehat{\Mix_{\bar{p},\bar{t}}}(S) =2^{-d} \sum_{j=0}^d p_j e^{-t_j |S|}.
\end{equation}
We further observe that the probability  that the $\Mix_{\bar{p},\bar{t}}$ distribution puts on the all-1's string is
\begin{equation}
\label{eq:Mixallones}
\Mix_{\bar{p},\bar{t}}(1^d)= 2^{-d} \sum_{j=0}^d p_j (1+e^{-t_j})^d.
\end{equation}
We now partially define the two distributions $P_1$ and $P_2$ in the statement of \Cref{thm:cube-lb}: Set
\begin{equation}\label{eq:choicegamma}
\gamma:=\frac{8\log(1/\eps)}{d}
\end{equation}
and set 
\begin{align} 
P_1 &:= \Mix_{\bar{\mu},\bar{t}}, \quad \text{where~}\bar{\mu}=(\mu_0,\mu_1,\dots,\mu_d),\ 
\bar{t} = (t_0,t_1,\dots,t_d) \text{~where~} t_j = j\gamma; \label{eq:P1}\\[0.6ex]
P_2 &:= \Mix_{\bar{\nu},\bar{t}}, \quad \text{where~}\bar{\nu}=(\nu_0,\nu_1,\dots,\nu_d) \text{~and as above~}
\bar{t} = (t_0,t_1,\dots,t_d) \text{~where~} t_j = j\gamma,\label{eq:P2}
\end{align}
where the mixing weights $\mu_0,\mu_1,\dots,\mu_d$ and $\nu_0,\nu_1,\dots,\nu_d$ will be determined in the next subsection, thus completing the definition of $P_1$ and $P_2$.

\subsubsection{Proof of \Cref{thm:cube-lb}}

To finalize $\gamma$ and the mixing weights in $\bar{\mu}$ and $\bar{\nu}$, let's understand what we need from them so that $P_1$ and $P_2$ satisfy conditions listed in \Cref{thm:cube-lb}.

By \Cref{eq:MixhatS} and \Cref{eq:P1,eq:P2} we have for each $S\subseteq \{1,\dots,d\}$ that
\begin{equation} \label{eq:Pcoeffs}
\widehat{P_1}(S) = 2^{-d} \sum_{j=0}^d \mu_j e^{-j|S|\gamma}\quad \text{and}\quad
\widehat{P_2}(S) = 2^{-d} \sum_{j=0}^d \nu_j e^{-j|S|\gamma}. 
\end{equation}
Desideratum (1) from the statement of \Cref{thm:cube-lb} translates into requiring that 
\begin{equation} \label{eq:cherry}
\abs{
\sum_{j=0}^d (\mu_j - \nu_j)\left(e^{-|S|\gamma}\right)^j}
\leq  2^{-\Omega\left(\sqrt{d \log(1/\eps)}\right)},
\end{equation}
for all $|S|=0,1,\dots,c\sqrt{d/\log(1/\eps)}$ for some constant $c>0$.
Let $A(x)$ denote the real polynomial $$A(x)=\sum_{j=0}^d a_jx^j,\quad \text{where each 
  $a_j=\mu_j-\nu_j$.}
  $$
Given that $$e^{-|S|\gamma}\ge 1-|S|\gamma
=1-8c\sqrt{\frac{\log(1/\eps)}{d}},
$$
the following condition implies 
  desideratum (1):
\begin{equation}\label{eq:gimmecherry}
\left|A(x)\right|\le 2^{\Omega\left(\sqrt{d\log(1/\eps)}\right)},\quad\text{for all}\  x\in 
\left[1-8c\sqrt{\frac{\log(1/\eps)}{d}},1\right].
\end{equation}



To obtain a condition that implies desideratum (2), we will use the following claim:
\begin{claim} \label{claim:A}
For any two vectors of mixing weights $\overline{\mu}$ and $\overline{\nu}$ (which define distributions $P_1$ and $P_2$, recalling \Cref{eq:P1,eq:P2}), we have that
$P_1(1^d)-P_2(1^d) = \mu_0 - \nu_0 \pm \eps.$
\end{claim}
\begin{proof}
By \Cref{eq:Mixallones} and \Cref{eq:P1,eq:P2}, we have
\[
P_1(1^d) - P_2(1^d) = \sum_{j=0}^d 2^{-d} \cdot (\mu_j-\nu_j) \cdot (1+e^{-j\gamma})^d.
\]
When $j=0$, the summand is simply $\mu_0 - \nu_0$.
To bound the rest of the sum, we claim that 
\begin{equation} \label{eq:epsclaim}
2^{-d} \sum_{j=1}^d (1 + e^{-j\gamma})^d \leq {\frac \eps 2}
\end{equation}
(we will prove this below).  Given this, 
since $\sum_{j=1}^d |\mu_j - \nu_j| \leq \sum_{j=1}^d \mu_j + \nu_j \leq 2$, we get that
\[
\abs{\sum_{j=1}^d 2^{-d} \cdot (\mu_j - \nu_j) \cdot (1 + e^{-j\gamma})^d}
\leq 2 \cdot {\frac \eps 2} = \eps,
\]
which gives \Cref{claim:A}.

To establish \Cref{eq:epsclaim} we break the sum into two parts.
For the first part, we have
\begin{align*}
2^{-d} \sum_{j=1}^{1/\gamma} (1 + e^{-j\gamma})^d \leq
2^{-d} \sum_{j=1}^{1/\gamma} \pbra{2 - {\frac {j\gamma} {2}}}^d < \sum_{j=1}^{1/\gamma} \pbra{1 - {\frac {j\gamma} {4}}}^d
\leq \sum_{j=1}^{1/\gamma} e^{-j\gamma d/4}
= \sum_{j=1}^{1/\gamma} \eps^{2j} < \frac{\eps}{4},
\end{align*}
where we used the choice of $\gamma$ from \Cref{eq:choicegamma} and the fact that $\eps < 1/170$.
For the rest, we have
\begin{align*}
2^{-d} \sum_{j>1/\gamma}^d (1+e^{-j\gamma})^d \leq 2^{-d} \cdot d\cdot  (1+e^{-1})^d < \frac{\eps}{4},
\end{align*}
using $\eps > 2^{-d/3}.$
Combining these two parts gives \Cref{eq:epsclaim} as desired.
\end{proof}

We use the following claim to finalize the mixing weights in $\bar{\mu}$ and $\bar{\nu}$:

\begin{claim} \label{claim:poly-exists}
There exists a real polynomial $A(x) = \sum_{j=0}^d a_j x^j$  such that 
\begin{itemize}
\item [(i)] $\sum_{j\in [0:d]\hspace{0.03cm}:\hspace{0.03cm} a_j \geq 0} a_j = \sum_{j \in [0:d]\hspace{0.03cm}:\hspace{0.03cm} a_j < 0} -a_j = 1$;
%
\item [(ii)] 
\Cref{eq:gimmecherry} holds (for some sufficiently small constant $c>0$); and
\item [(iii)] $|a_0| \geq 3 \eps$.
\end{itemize}
\end{claim}

Before proving \Cref{claim:poly-exists} we explain why it completes the proof of \Cref{thm:cube-lb}.
Given such a polynomial $A(x)$, 
for each $j \in \{0,1,\dots,d\}$, we set the mixing weights as follows: 
\begin{align*}
\mu_j = a_j\ \text{and}\ \nu_j = 0 \ \text{if $a_j \geq 0$;}\quad 
\mu_j = 0\ \text{and}\ \nu_j =  -a_j\  \text{if $a_j < 0$.}
\end{align*}
These $\mu_j, \nu_j$ values are non-negative, satisfy $a_j = \mu_j-\nu_j$ for each $j$, and satisfy $\sum_{j} \mu_j = \sum_{j} \nu_j = 1$.
As shown above (ii) (i.e.~\Cref{eq:gimmecherry}) implies desideratum (1) of \Cref{thm:cube-lb}.
Finally, by \Cref{claim:A}, if $|\mu_0 - \nu_0| \geq 3\eps,$ then $|P_1(1^d) - P_2(1^d)| \geq 2\eps$ giving desideratum (2) of \Cref{thm:cube-lb}.

\subsection{Proof of \Cref{claim:poly-exists}}\label{sec:proofclaim}

We use the polynomial that is given by the following result of Erd\'{e}lyi:

\begin{lemma} [Lemma~3.3 of \cite{Erdelyi16}] \label{thm:Erdelyi}
For any $L \in (0,1/17)$ and any $n \in \N$, there exists a~real-coefficient polynomial $A(x)=\sum_{j=0}^n a_j x^j$ with $|a_0| \geq L \cdot \sum_{j=1}^n |a_j|$ such that $A$ has at least
\[
\min\cbra{{\frac 2 7} \sqrt{n \cdot (-\ln L)},n}
\]
repeated roots at 1.
\end{lemma}

We apply \Cref{thm:Erdelyi} with $n=d$ and $L=10\eps$ (recall that $\eps<1/170$) to obtain the polynomial $\smash{A(x)=\sum_{j=0}^d a_jx^j}$ for \Cref{claim:poly-exists}.
Note that by scaling we may further assume that $\sum_{j=0}^d |a_j|=2$. 
With these choices we have that $|a_0| \geq 10 \eps (2 - |a_0|)$, hence $|a_0| \geq 3 \eps$ (so property (iii) of \Cref{claim:poly-exists} holds). On the other hand, since $A$ has a root at 1, we have
$\sum_{j=0}^d a_j = 0,$ which together with $\smash{\sum_{j=0}^d |a_j|=2}$ gives us property (i) of \Cref{claim:poly-exists}.

For property (ii) of \Cref{claim:poly-exists}, we will use the following result of Borwein, Erd\'{e}lyi and  K\'{o}s,
and the fact that $A(x)$ has at least 
$\Omega(\sqrt{d\log(1/\eps)})$ repeated roots at 1 (where we used $\eps>2^{-d/3})$):

\begin{lemma} [Lemma~5.4 of \cite{BEK:97}] \label{lem:BEK}
Let $B: \C \to \C$ be defined as $B(x) = \sum_{j=0}^n b_j x^j$ satisfying $|b_j| \leq 1$ for all $0 \leq j \leq n$. If $B$ have $k$ repeated roots at 1, then we have
\[
\sup_{x \in I} \big|B(x)\big| \leq (n+1)\pbra{{\frac e 9}}^k,\quad\text{where $I=\left[1-\frac{k}{9n},1\right]$.}
\]
\end{lemma}
Taking $B(x)=A(x)/2$ we have that $|b_j| \leq 1$ for all $j$, and $B$ has 
$$k \geq {\frac 2 7} \sqrt{d\ln \left(\frac 1 {10 \eps}\right)}$$ repeated roots at 1. For a suitably small positive choice of the absolute constant $c$, we have that $$8c\sqrt{\frac {\log(1/\eps)}{d}} \leq \frac{k}{9d},$$  so we may apply \Cref{lem:BEK} and we get that for any $x$ in the interval in \Cref{eq:gimmecherry}, we have
\[
\big|A(x)\big| = 2 \cdot 
\big|B(x)\big| 
\leq 2^{-\Theta(k)}
= 2^{-\Omega\left(\sqrt{d \log(1/\eps)}\right)}.  
\]
This gives property (ii) and concludes the proof of \Cref{claim:poly-exists}.
 
\subsection{From the discrete cube to the $\mathbb{T}^d$ torus}\label{subsec:torus}

We now use the pair of distributions $P_1,P_2$ over $\bits^d$ from \Cref{sec:lb-discrete-cube} to give a pair of distributions $D_1,D_2$ over the torus $\mathbb{T}^d$ such that all low-degree Fourier coefficients are close to matching, yet the heavy hitter distance between $D_1$ and $D_2$ is large, thus proving \Cref{thm:De-distance-lower}.  For ease of reference we restate \Cref{thm:De-distance-lower}:

\Dedistancelower*

\begin{proof}
Let us define the following simple embedding of $\bits^d$ into $\mathbb{T}^d$: For each $x \in \bits^d$,
\[
\Emb(x) = \pbra{{\frac {1-x_1} 4},\dots,{\frac {1 - x_d} 4}},
\]
so $\Emb(x)$ ranges over $\{0,1/2\}^d \subset \mathbb{T}^d$ and in particular we have $\Emb(1^d) = (0,\dots,0).$

We also define the following mapping from $\Z^d$ to subsets of $\{1,\dots,d\}$ which will be useful for translating between Fourier coefficients over $\mathbb{T}^d$ and Fourier coefficients over $\bits^d$:
For each $\ell \in \Z^d$, 
\[
 \Par(\ell) = \big\{j \in \{1,\dots,d\}: (-1)^{\ell_j}=-1\big\}\quad\text{or equivalently,}\quad\big\{j \in \{1,\dots,d\}: \ell_j \text{~is odd}\big\}.
\]

For $b \in \{1,2\}$ and $z \in \bits^d$, recall that $P_b(z)$ is the amount of mass that the distribution $P_b$ puts on $z$.
We may write the distribution $P_b$ as 
\[
P_b(x) = \sum_{z \in \bits^d} P_{b}(z) \delta_z(x).
\]
The distribution $D_b$ is defined as
\begin{equation} \label{eq:Db}
D_b(y) := \sum_{z \in \bits^d} P_{b}(z) \delta_{\Emb(z)}(y),
\end{equation}
so the support of $D_b$ (a distribution over $\mathbb{T}^d$) is contained in $\{0,{1/2}\}^d$, and for $y \in \{0,{1/2}\}^d$, the amount of weight that $D_b$ puts on $y$ is equal to the amount of weight that $P_b$ puts on the unique string $z = \Emb^{-1}(y) \in \bits^d.$
From \Cref{eq:Db} we have that for $\ell \in \Z^d$, the Fourier coefficient 
\begin{align}
\widehat{D_b}(\ell) &= 
\sum_{z \in \bits^d} P_b(z) e^{2 \pi i (\ell \cdot \Emb(z))} \nonumber \\
&= \sum_{z \in \bits^d} P_b(z) \chi_{\Par(\ell)}(z) \label{eq:chch}\\
&= 2^d \cdot \Ex_{\bz \sim \bits^d} \sbra{P_b(\bz) \chi_{\Par(\ell)}(\bz)} = 2^d \cdot \widehat{P_b}(\Par(\ell)),\label{eq:dal}
\end{align}
where \Cref{eq:chch} holds because if $\ell_j$ is even then $e^{2\pi i \ell_j(1-z_j)/4}=1$ for both $z_j\in \{\pm 1\}$, and if $\ell_j$ is odd then $e^{2 \pi i \ell_j (1-z_j)/4}=z_j$ for $z_j \in \bits$.



Fix any $\ell \in \Z^d$ such that $$\|\ell_1\|_1 \leq c\sqrt{\frac{d}{\log(1/\eps)}},$$
where $c>0$ is the absolute constant in \Cref{thm:cube-lb}. Since $|\Par(\ell)| \leq \|\ell\|_1$,
by \Cref{eq:dal} and property (1) of \Cref{thm:cube-lb}, we have that
\[
\left|\widehat{D_1}(\ell) - \widehat{D_2}(\ell)\right|
= 2^d \cdot \left|\widehat{P_1}(\Par(\ell)) - \widehat{P_2}(\Par(\ell))\right|
\leq 2^{-\Omega\left(\sqrt{d \log(1/\eps)}\right)},
\]
establishing property (a) of \Cref{thm:De-distance-lower}.

\medskip

It remains to establish property (b) of \Cref{thm:De-distance-lower}, i.e.~to show that $\dDe^{(\ed)}(D_1,D_2)>\eps$.  This is simple:  let $x$ be the point $0^d \in \mathbb{T}^d$ and let $\tau = 0$.  We have that $D_1(x) - D_2(x) = P_1(1^d) - P_2(1^d)$, which is at least $2\eps$ by part (2) of \Cref{thm:cube-lb}. Moreover, $D_1(\Ball(x,\tau)) = D_1(x) = P_1(1^d)$, and $D_2(\Ball(x,\tau+\ed))=D_2(\Ball(x,\ed)) = D_2(x)$ (since $x$ is the only point in the support of $D_2$ that lies in $\Ball(x,\tau+\ed)$, recalling that $\ed <1/2$), which equals $P_2(1^d)$. So $D_1(\Ball(x,\tau)) - D_2(\Ball(x,\ed)) \geq 2\eps$
and hence indeed $\dDe^{(\ed)}(D_1,D_2)>\eps$.
This concludes the proof of \Cref{thm:De-distance-lower}.
\end{proof}

\section*{Acknowledgements}

The authors would like to thank Ankur Moitra for helpful discussions. %
X.C.~is supported by NSF grants CCF-2106429 and CCF-2107187. 
A.D.~is supported by NSF grant CCF 2045128. 
Y.H.~is supported by NSF grants CCF-2211238, CCF-2106429, and CCF-2238221
R.A.S.~is supported by NSF grants CCF-2211238 and CCF-2106429. 
T.Y.~is supported by NSF grants CCF-2211238, CCF-2106429, and AF-Medium 2212136.
T.Y.~and Y.H.~are also supported by an Amazon Research Award, Google CyberNYC award, and NSF grant CCF-2312242.

\bibliography{refs.bib}
\bibliographystyle{alpha}

\appendix

\crefname{appendix}{appendix}{appendices}
\Crefname{appendix}{Appendix}{Appendices}
\crefalias{section}{appendix}
\crefalias{subsection}{appendix}
\crefalias{subsubsection}{appendix}


\section{Small Wasserstein distance implies small heavy hitter distance} 
\label{ap:HH-Wasserstein}

\begin{proposition} \label{prop:HH-Wasserstein}
Fix any scale parameter $0 < \ed \leq 1$, and let $D_1,D_2$ be two distributions over $\mathbb{T}^d$ satisfying $\dW(D_1,D_2) \leq \eps\cdot \ed$.  Then $\dDe^{(\ed)}(D_1,D_2) \leq \eps.$
\end{proposition}
\begin{proof}
We show the contrapositive; so suppose that $\dDe^{(\ed)}(D_1,D_2) > \eps.$
This means that there exists some $x \in \mathbb{T}^d$ and some $\tau \in [0,\ed]$ such that
$D(\Ball(x,\tau)) > D'(\Ball(x,\tau+\ed)) + \eps$, where either $D=D_1,D'=D_2$ or vice versa.  
Hence $\dW(D,D') \geq \eps \cdot \ed$, since at least an $\eps$ amount of mass under $D$ must be moved a distance of at least $\ed$ (to take it from lying inside the radius-$\tau$ ball around $x$ to lying outside the radius-$(\tau+\ed)$ ball around $x$) in order to transform $D$ into $D'$. 
\end{proof}


\section{Wasserstein reconstruction}
\label{sec:appendix-wasserstein}

In this section, we provide proofs of the results stated in~\Cref{sec:multi-dim-Wasserstein}.  
    
\subsection{Upper bound}
\label{sec:d-dim-ub-structural}

We begin by restating \Cref{thm:multi-dim-Wass-ub} in a slightly more general setting. 
In particular, our result holds for general \emph{signed measures of finite total variation} on~$\mathbb{T}^d$, normalized so that their total variation is~$1$, rather than only for probability distributions. 
We refer to such objects as \emph{signals}. 

Equivalently, with a mild (and standard) abuse of notation, we identify such signals with~$L^1$-integrable functions $f : \mathbb{T}^d \to \mathbb{R}$ satisfying
\[
    \int_{\mathbb{T}^d} |f(x)|\,dx = 1,
\]
while allowing discrete components such as Dirac deltas. 
(This convention lets us treat continuous and discrete signals in a unified manner.)  
Note that non-negative normalized signals correspond to probability distributions over~$\mathbb{T}^d$, which is the setting considered in the main body of the paper. 
For background on finite signed measures and total variation, see, e.g., Chapter~6 of~\cite{rudin1987real} or Chapter~3 of~\cite{folland1999real}.

\paragraph{Wasserstein distance between signals.}

The characterization of Wasserstein distance from \Cref{subsec:prelims-wass} extends naturally to this more general setting. 
Given two signals $f, g : \mathbb{T}^d \to \mathbb{R}$ with
\[
    \int_{\mathbb{T}^d} |f(x)|\,dx = \int_{\mathbb{T}^d} |g(x)|\,dx = 1,
\]
we define their Wasserstein distance as in \Cref{eq:KR}:
\[
    \dW(f,g) := \sup_{h: \mathbb{T}^d \to \mathbb{R}} \int h(x)\cdot (f(x) - g(x))\,dx,
\]
where the supremum is taken over all $1$-Lipschitz functions $h$ on $\mathbb{T}^d$ satisfying $\|h\|_\infty \le \sqrt{d}$.  
(We remark that in this more general context in which $f$ and $g$ need not integrate to the same value, the condition that $\|h\|_\infty \leq \sqrt{d}$ is essential for the definition to make sense.)
This notion has been previously considered in the literature, see e.g.~\cite{Hanin99,PRS23}.

\paragraph{Wasserstein reconstruction for signals.} 

We can now state the generalization of \Cref{thm:multi-dim-Wass-ub} for signals:  

\begin{theorem}[Generalization of \Cref{thm:multi-dim-Wass-ub}]~\label{thm:multi-dim-Wass-ub-signals}
    Let $f_1, f_2$ be two signals over the the $d$-dimensional torus $\mathbb{T}^d$. For any $0 < \eps < 1/2$, and any $d \geq 1$, 
    \[
    \text{let~}T = \left\lceil 6\sqrt{d}/\eps \right\rceil,
    \quad \text{and let~} 
    \kappa = \begin{cases}
        0.001 \eps / \log(1/\eps) & \text{~when~}d = 1\\
        \left(0.01 \eps / \sqrt{d}\right)^d & \text{~when~}d \ge 2
        \end{cases}.
    \]
        If $f_1,f_2$ satisfy $|\widehat{f_1}(\ell) - \widehat{f_2}(\ell) | \le \kappa$ for all integer vectors $\ell = (\ell_1, \dots, \ell_d) \in \mathbb{Z}^d$ with $\left\| \ell \right\|_{\infty} \le T$, then $\dW(f_1, f_2) \le \eps$.
\end{theorem}

We note that this theorem holds even for generalized signals $f_1,f_2$ whose absolute values need not integrate to $1$. This will be useful for the algorithmic result given in \Cref{sec:wasserstein-algo}.

\begin{remark}[Heavy hitter recovery for signals]
	It is natural to ask whether heavy hitter recovery extends to the more general setting of (possibly signed) measures rather than just distributions (i.e., non-negative measures). 
	Unfortunately, this is not the case: non-negativity is essential for heavy hitter recovery in the absence of separation assumptions. 
	Indeed, if signed signals are allowed and there is no minimal separation assumption, then
	it is impossible to distinguish the absence of a spike at $x_0$ from the two-spike signed signal  
	\[
	    {\frac 1 2} \cdot \delta_{x_0} - {\frac 1 2} \cdot \delta_{x_0 + \Delta}, 
	\]
	(where $\Delta$ may be arbitrarily small) using only Fourier measurements. 
	In particular, this means that our heavy hitter objective of ``recovering the spikes'' cannot be achieved. 
	Since the intuitive goal of heavy hitter recovery is to ``find the spikes,'' i.e.,~identify regions where the signal's mass is concentrated, it is most natural to study this problem for nonnegative signals, where such regions correspond to high-density areas of a distribution. 
	
	We thus use $D$ to denote such non-negative signals, which we interpret as distributions (since $\int |D| = \int D = 1$), and reserve $f$ for general, possibly signed, signals. 
\end{remark}

\subsubsection{Proof idea}
\label{subsec:appendix-wass-tech}

Restating \Cref{thm:multi-dim-Wass-ub-signals}, we wish to show that if the Wasserstein distance between two signals $f_1$ and $f_2$ is large, say $\dW(f_1, f_2) \ge \epsilon$, then there must be a Fourier coefficient $\ell$ (with $\Vert \ell \Vert_\infty \le T$) such that $|\widehat{f}_1(\ell) - \widehat{f}_2(\ell)|>\kappa$.   
Establishing this immediately implies that if the Fourier coefficients $\{\widetilde{f}(\ell)\}_{\Vert \ell \Vert_\infty \le T}$ are known up to additive error $\kappa$, one can recover $f$ within Wasserstein distance $\eps$. 
To establish this, we follow a standard analytic approach using a mollifier as described in \Cref{subsubsec:HH-tech-intro}.   

Assuming for now a suitable mollifier $J$ that is ``$\eta$-concentrated'' in the sense that  
\[
    \Ex_{\bx \sim J}[\dtor(\bx, 0^d)] = \eta\,,
\]
we define  
\[
    f_1':= f_1 \ast J \qquad\text{and}\qquad f_2' := f_2 \ast J\,.
\]
Since $J$ is $\eta$-concentrated around $0^d$ (as described above), it is straightforward to verify that if $\eta \le \epsilon/3$, we will have $\dW(f_1', f_2') \ge \epsilon/3$. 
On the Fourier side, convolution corresponds to pointwise multiplication: 
\[
    \widehat{f_1'} (\ell) = \widehat{f_1}(\ell) \cdot \widehat{J}(\ell) 
    \qquad\text{and}\qquad 
    \widehat{f_2'} (\ell) = \widehat{f_2}(\ell) \cdot \widehat{J}(\ell)
\]
Our argument further uses three more simple but crucial ingredients from Fourier analysis:
\begin{itemize}
    \item Plancherel's identity, which allows us to relate the inner product of two signals to the inner product of their Fourier transforms;
    \item The bandlimitedness of the mollifier $J$, which ensures that $\widehat{f_1'}(\ell) = \widehat{f_2'}(\ell)=0$ for all $\Vert \ell \Vert_\infty >T$; and
    \item The decay of Fourier coefficients for Lipschitz functions, namely that $|\widehat{g}(\ell)| = O(\Vert \ell \Vert_2^{-1})$ for Lipschitz $g$ (\Cref{fact:Lipschitz-Fourier}).
\end{itemize}

Using the three ingredients above, one can show that a lower bound on $\dW(f_1', f_2')$ necessarily implies the existence of a frequency $\ell$ with $\|\ell\|_\infty \leq T$ for which $|\widehat{f_1}(\ell) - \widehat{f_2}(\ell)|$ is large. 
The resulting quantitative bounds depend on the choice of mollifier $J$. 
For our analysis, we employ \emph{Jackson's kernel} $J_{d,n}$ (in $d$ dimensions, parameterized by the integer $n$; see \Cref{def:jackson-kernel} for the precise definition).
Substituting this kernel into the framework outlined above yields the quantitative guarantees stated in \Cref{thm:multi-dim-Wass-ub-signals}.

\subsubsection{Jackson's kernel}\label{sec:prelim-jackson}

As mentioned above, our proof of \Cref{thm:multi-dim-Wass-ub-signals} relies on Jackson's kernel, which we now define. 
The reader may safely skip ahead to \Cref{subsec:appendix-wass-ub-proof}, where we prove \Cref{thm:multi-dim-Wass-ub-signals}, and return here as needed. 

\begin{definition}[Jackson's kernel] \label{def:jackson-kernel}
    For any integer $n \ge 1$, the $1$-dimensional \emph{Jackson's kernel} over the torus $\mathbb{T}$ is defined by
    \begin{equation*}
        J_n(x) := \alpha_n \cdot \frac{\sin^4(\pi nx)}{\sin^4(\pi x)}\text{\qquad(and naturally $J_n(0) := \alpha_nn^4$)},
    \end{equation*}
    where $\alpha_n$ is a constant given by
    \begin{equation*}
        \alpha_n := \frac{3}{n (2n^2 + 1)}.
    \end{equation*}
    The $d$-dimensional Jackson's kernel over $\mathbb{T}^d$ is defined by
    \begin{equation*}
        J_{d,n}(x_1, \dots, x_d) := \prod_{i=1}^d J_n(x_i) = \alpha_n^d \cdot \prod_{i=1}^d \frac{\sin^4(\pi nx_i)}{\sin^4(\pi x_i)}.
    \end{equation*}
\end{definition}

Note that the function $J_n(x)$ viewed as a function over $\R$ is periodic. 
The following proposition shows that Jackson's kernel defined above integrates to $1$, so the $J_{d,n}$ defines a distribution over $\mathbb{T}^d$. 

\begin{restatable}{proposition}{jacksonnormalization}\label{prop:jackson-normalization}
    We have 
    \begin{equation*}
        \int_{x \in \mathbb{T}} J_n(x)\, dx = 1\,.
    \end{equation*}
\end{restatable}

\begin{proof}
The proof here is adapted from \cite{Giapitzakis-stackex}. We have
\begin{equation} 
\label{eq:gumbo}\frac{\sin(\pi nx)}{\sin(\pi x)} = \frac{e^{\pi inx} - e^{-\pi inx}}{e^{\pi ix} - e^{-\pi ix}} = e^{-\pi i(n - 1)x} \cdot \frac{e^{2\pi inx} - 1}{e^{2\pi ix} - 1} = e^{-\pi i(n - 1)x}\sum_{j = 0}^{n - 1}e^{2\pi ijx},\end{equation}
so
\[\frac{\sin^4(\pi nx)}{\sin^4(\pi x)} = e^{-4\pi i(n - 1)x}\sum_{j_1, j_2, j_3, j_4 \in \{0, \dots, n - 1\}}e^{2\pi i(j_1 + j_2 + j_3 + j_4)x}.\]
Note that for $t \in \mathbb{Z}$ we have $\int_{x \in \mathbb{T}}e^{2\pi itx}dx = \begin{cases}
    1 & \text{if~}t = 0\\
    0 & \text{if~}t \ne 0
\end{cases}$, so we thus have
\[\int_{\mathbb{T}}\frac{\sin^4(\pi nx)}{\sin^4(\pi x)}dx = \left|\{(j_1, j_2, j_3, j_4) \in \{0, \dots, n - 1\}^4 \mid j_1 + j_2 + j_3 + j_4 = 2n - 2\}\right|.\]
The size of the set above equals the number of $(j_1, j_2, j_3, j_4) \in \mathbb{Z}_{\ge 0}^4$ such that $j_1 + j_2 + j_3 + j_4 = 2n - 2$, minus the number of $(j_1, j_2, j_3, j_4) \in \mathbb{Z}^4_{\ge 0}$ such that $j_1 + j_2 + j_3 + j_4 = 2n - 2$ and one of $j_1, j_2, j_3, j_4$ is at least $n$ (there cannot be more than one of them that is at least $n$ since that would make the sum at least $2n$). Thus 
\[\int_{\mathbb{T}}\frac{\sin^4(\pi nx)}{\sin^4(\pi x)} = \binom{2n + 1}3 - 4\binom{n + 1}3 = \frac{2n^3 + n}3 = \frac1{\alpha_n}. \qedhere\]
\end{proof}

The next lemma bounds the expected $\ell_2$-norm of a random $\bx$ under the distribution defined by $J_{d,n}$.

\begin{restatable}{lemma} {jacksonexpecteddistance}\label{lmm:jackson-expected-distance}
    For any $d, n \ge 1$,
    \begin{equation*}
        \Ex_{\bx \sim J_{d,n}} \left[ 
        \dtor(\bx,0^d)
        \right] \le \frac{\sqrt{d}}{n}.
    \end{equation*}
\end{restatable}

\begin{proof}
The proof here uses ideas from \cite[Lemma 9.6]{Shadrin05-partIIIat}.
\begin{align*}
    \Ex_{\bx \sim J_{d,n}} \sbra{\dtor(\bx,0^d)}
    &\le \sqrt{\Ex_{\bx \sim J_{d,n}}\sbra{\dtor(\bx,0^d)^2 }}\\
    &= \sqrt{2^d\int_{[0, 1/2]^d}\|x\|_2^2J_{d, n}(x)dx}\tag{by symmetry}\\
    &=\sqrt{2^d\int_{[0, 1/2]^d}(x_1^2 + \cdots + x_d^2)J_n(x_1) \cdots J_n(x_d)dx_1 \cdots dx_d} \\
    &=\sqrt{2^dd\int_{[0, 1/2]^d}x_1^2J_n(x_1) \cdots J_n(x_d)dx_1 \cdots dx_d}\tag{by symmetry}\\
    &=\sqrt{2d\int_0^{1/2}x_1^2J_n(x_1)dx_1}\tag{by \Cref{prop:jackson-normalization}}\\
    &\le \sqrt{\frac{d\alpha_n}{8}\int_0^{1/2}x_1^2 \left(\frac{\sin(\pi nx_1)}{x_1}\right)^4dx_1}\tag{$\sin(\pi x_1) \ge 2x_1$ for $0 \le x_1 \le 1/2$}\\
    &\le \sqrt{\frac{\pi nd\alpha_n}{8}\int_0^{\pi n/2}u^2\left(\frac{\sin u}u\right)^4du}\tag{substitute $u := \pi nx_1$}\\
    &\le \sqrt{\frac{\pi nd\alpha_n}{8}\left(\int_0^{1}u^2du + \int_{1}^{+\infty}u^{-2}du\right)}\tag{$\sin u \le \min\{u, 1\}$ for $u \ge 0$}\\
    &\le \frac{\sqrt{\pi}}2 \cdot \frac{\sqrt{d}}n\,,\tag{$\alpha_n \le 3/(2n^3)$}
\end{align*}
which completes the proof. 
\end{proof}

We also need the following lemma about the Fourier coefficients of Jackson's kernel.

\begin{restatable}{lemma}{jacksonfourier}
\label{lmm:jackson-fourier}
    For any $d, n \ge 1$ and $\ell \in \mathbb{Z}^d$, $\left\lvert \widehat{J_{d,n}}(\ell) \right\rvert \le 1$. Moreover, $\widehat{J_{d,n}}(\ell) = 0$ if $\left\| \ell \right\|_{\infty} \ge 2n$.
\end{restatable}

\begin{proof}
Note that
\[\widehat{J_{d, n}}(\ell) = \int_{x \in \mathbb{T}^d}e^{2\pi i(\ell \cdot x)}J_{d, n}(x)dx = \prod_{k = 1}^d\left(\int_{x_k \in \mathbb{T}}e^{2\pi i\ell_kx_k}J_n(x_k)dx_k\right) = \prod_{k = 1}^d\widehat{J_{1, n}}(\ell_k),\]
so we only need to consider the case where $d = 1$. When $d = 1$,
\begin{align*}
    \widehat{J_{1, n}}(\ell)
    &= \int_{\mathbb{T}}e^{2\pi i\ell x}\alpha_n\cdot \frac{\sin^4(\pi nx)}{\sin^4(\pi x)}dx\\
    &= \alpha_n\int_{\mathbb{T}}e^{2\pi i(\ell - 2n + 2)x} \cdot \left(\frac{e^{2\pi nx} - 1}{e^{2\pi x} - 1}\right)^4dx
    \tag{recalling \Cref{eq:gumbo}}\\
    &= \alpha_n\int_{\mathbb{T}}e^{2\pi i(\ell - 2n + 2)x} \cdot \left(\sum_{j = 0}^{n - 1}e^{2\pi jx}\right)^4dx\\
    &= \alpha_n\int_{\mathbb{T}}\sum_{j_1, j_2, j_3, j_4 \in \{0, \dots, n - 1\}}e^{2\pi i(\ell - 2n + 2 + j_1 + j_2 + j_3 + j_4)x}dx\\
    &= \alpha_n \cdot \left|\{(j_1, j_2, j_3, j_4) \in \{0, \dots, n - 1\}^4 \mid j_1 + j_2 + j_3 + j_4 = 2n - 2 - \ell\}\right|.
\end{align*}
The size of the set above can be calculated by the number of $(j_1, j_2, j_3, j_4) \in \mathbb{Z}_{\ge 0}^4$ such that $j_1 + j_2 + j_3 + j_4 = 2n - 2 - \ell$, minus the number of $(j_1, j_2, j_3, j_4) \in \mathbb{Z}_{\ge 0}$ such that $j_1 + j_2 + j_3 + j_4 = 2n - 2 - \ell$ and one of $j_1, j_2, j_3, j_4$ is at least $n$ (there cannot be more than one of them that is at least $n$ since that would make the sum at least $2n$).
Thus
\[\widehat{J_{1, n}}(\ell) = \begin{cases}
    \alpha_n\left(\binom{2n + 1 - \ell}3 - 4\binom{n + 1 - \ell}3\right) & 0 \le \ell \le n - 2\\
    \alpha_n\binom{2n + 1 - \ell}3 & n - 1 \le \ell \le 2n - 2\\
    0 & \ell \ge 2n - 1
\end{cases},\]
which can be straightforwardly verified to be at most 1 since $\alpha_n = \frac{3}{n (2n^2 + 1)}.$ 
\end{proof}

We also use the fact that Jackson's kernel is Lipschitz:

\begin{restatable}{lemma}{jacksonlipschitz}
\label{lmm:jackson-lipschitz}
    For any $d, n \ge 1$, $J_{d,n}$ is $(3\pi n^2(3n/2)^{d - 1}\sqrt d)$-Lipschitz.
\end{restatable}

\begin{proof}
We first observe that for $0 < x < 1$,
\begin{equation} \label{eq:useme}\left|\frac{\sin{\pi nx}}{\sin{\pi x}}\right| = \left|\frac{e^{\pi i nx} - e^{-\pi i nx}}{e^{\pi i x} - e^{-\pi i x}}\right| = \left|e^{\pi i(n - 1)x} + e^{\pi i(n - 3)x} + \cdots + e^{\pi i(-n + 1)x}\right| \le n.\end{equation}

Next, we show that $0 \le J_n(x) \le 3n/2$ for all $x \in \mathbb{T}$. The lower bound is trivial. For the upper bound, when $x = 0$, $J_n(x) = \alpha_nn^4 \le 3n/2$; when $0 < x < 1$, thanks to \Cref{eq:useme} we also have that $J_n(x) \le \alpha_nn^4 \le 3n/2$.

We now show that $|J'_n(x)|$ is bounded by $3\pi n^2$ for every $x \in \mathbb{T}$. Note that the Taylor series of $\sin t$ at $t = 0$ is $t + O(t^3)$, so we have
\[J_n'(0) = \lim_{x \to 0}\frac1x \cdot \left(\alpha_nn^4 - \frac{\alpha_n\sin^4(\pi nx)}{\sin^4(\pi x)}\right) = \alpha_n\lim_{x \to 0}\frac{n^4(\pi^4x^4 + O(x^6)) - (\pi^4n^4x^4 + O(x^6))}{x(\pi^4x^4 + O(x^6))} = 0.\]
For $0 < x < 1$, we have
\begin{align*}
\left|\left(\frac{\sin(\pi n x)}{\sin(\pi x)}\right)'\right|
&= \left|\left(\sum_{j = 0}^{n - 1} e^{\pi i(2j - n + 1)x}\right)'\right| \\
&= \left|\sum_{j = 0}^{n - 1} \pi i (2j - n + 1) e^{\pi i(2j - n + 1)x}\right| \\
&\le \pi \sum_{j = 0}^{n - 1} |2j - n + 1| \le \frac{\pi n^2}{2}\,,
\end{align*}
so
\[\left|J_n'(x)\right| = \left|4\alpha_n \cdot \left(\frac{\sin(\pi nx)}{\sin(\pi x)}\right)' \cdot \left(\frac{\sin(\pi nx)}{\sin(\pi x)}\right)^3\right| \le \frac6{n^3} \cdot \frac{\pi n^2}2 \cdot n^3 = 3\pi n^2.\]

It follows that
\[\left|\frac{\partial J_{d, n}(y)}{\partial y_k}\right| = \left|J_n'(y_k) \cdot \prod_{1 \le j \le d, j \ne k}J_n(y_j)\right| \le 3\pi n^2 \cdot \left(\frac{3n}2\right)^{d - 1},\]
so by Cauchy-Schwarz, the function $J_{d, n}$ is $(3\pi n^2(3n/2)^{d - 1}\sqrt d)$-Lipschitz.
\end{proof}

\subsubsection{Proof of \Cref{thm:multi-dim-Wass-ub}}
\label{subsec:appendix-wass-ub-proof}

    Towards a contradiction, we assume that $f_1$ and $f_2$ are two signals such that $\dW(f_1, f_2) > \eps$. 
    The idea is to smooth each signal using Jackson's kernel. Define $\widetilde{f_1} := f_1 * J_{d,n}$ and $\widetilde{f_2} := f_2 * J_{d,n}$, where $n = \lceil T/2 \rceil$, $J_{d,n}$ is Jackson's kernel in  \Cref{def:jackson-kernel}, and $*$ denotes convolution. By \Cref{lmm:jackson-expected-distance}, $\dW(f_i, \widetilde{f_i}) \le \sqrt{d}/n \le \eps/3$ for $i = 1, 2$, so by the triangle inequality we have that $\dW(\widetilde{f_1}, \widetilde{f_2}) \ge \eps/3$.
    This means that there exists some $1$-Lipschitz function $h: \mathbb{T}^d \to \mathbb{R}$ with $\|h\|_{\infty} \le \sqrt{d}$ satisfying
    \begin{equation*}
        \int_{\mathbb{T}^d} \left(\widetilde{f_1}(x) - \widetilde{f_2}(x)\right) \cdot h(x) ~ dx \ge \eps/3.
    \end{equation*}
    Recalling that each $\widetilde{f_i}$ was obtained from $f_i$ (which may not be square-integrable) by convolving with Jackson's kernel, we see that each $\widetilde{f_i}$ is $L^2$-integrable, so we may apply Plancherel's theorem to restate the above inequality as
    \begin{equation*}
        \sum_{\ell \in \mathbb{Z}^d} \left( \widehat{\widetilde{f_1}}(\ell) - \widehat{\widetilde{f_2}}(\ell) \right) \cdot \overline{\widehat{h}(\ell)} \ge \eps/3.
    \end{equation*}

    By the definition of $\widetilde{f_1}$ and $\widetilde{f_2}$ and the definition of convolution, we have that 
    \[
    	\widehat{\widetilde{f_1}}(\ell) - \widehat{\widetilde{f_2}}(\ell) = \left( \widehat{f_1}(\ell) - \widehat{f_2}(\ell) \right) \cdot \widehat{J_{d,n}}(\ell)\,.
    \] 
    However, by \Cref{lmm:jackson-fourier}, we know that $\widehat{J_{d,n}}(\ell) = 0$ whenever $\|\ell\|_{\infty} \ge 2n$.
        Also, when $\|\ell\|_{\infty} = 0$, we have $| \widehat{h}(0^d) | \le \sqrt{d}$ from $\| h \|_{\infty} \le \sqrt{d}$. Combining this with the assumption that $| \widehat{f_1}(0^d) - \widehat{f_2}(0^d) | \le \kappa$ and the choice of $\kappa$, it is easy to see that 
        \[\left( \widehat{\widetilde{f_1}}(0^d) - \widehat{\widetilde{f_2}}(0^d) \right) \cdot \overline{\widehat{h}(0^d)} \le \frac{\eps}{6}\,.\]
    Hence,
    \begin{equation} \label{eq:pf-thm-wasserstein-d-dim-ub:sum-ell-2n-ge-eps-6}
        \sum_{0 < \|\ell\|_{\infty} \le 2n-1} \left( \widehat{f_1}(\ell) - \widehat{f_2}(\ell) \right) \cdot \widehat{J_{d,n}}(\ell) \cdot \overline{\widehat{h}(\ell)} \ge \frac{\eps}{6}.
    \end{equation}

    On the other hand, since $h$ is $1$-Lipschitz, for $0^d \neq \ell \in  \Z^d$ by \Cref{fact:Lipschitz-Fourier} we have that $\left\lvert\widehat{h}(\ell)\right\rvert \le 1 / \|\ell\|_2$. \Cref{lmm:jackson-fourier} tells us that $\left\lvert \widehat{J_{d,n}}(\ell) \right\rvert \le 1$. Since by assumption we have that $\left\vert \widehat{f_1}(\ell) - \widehat{f_2}(\ell) \right\vert \le \kappa$ for all $\ell$ with $\left\| \ell \right\|_{\infty} \le 2n-1 \le T$, we have
    \begin{align*}
        \sum_{0 < \|\ell\|_{\infty} \le 2n-1} \left( \widehat{f_1}(\ell) - \widehat{f_2}(\ell) \right) \cdot \widehat{J_{d,n}}(\ell) \cdot \overline{\widehat{h}(\ell)} 
        \le & \sum_{0 < \|\ell\|_{\infty} \le 2n-1} | \widehat{f_1}(\ell) - \widehat{f_2}(\ell) | \cdot | \widehat{J_{d,n}}(\ell) | \cdot | \overline{\widehat{h}(\ell)} | \\
        \le & \kappa \cdot \sum_{0 < \|\ell\|_{\infty} \le 2n-1} \frac{1}{\|\ell\|_2} \\
        \le & \kappa \cdot \sum_{0 < \|\ell\|_{\infty} \le 2n-1} \frac{1}{\|\ell\|_{\infty}} \\
        \le & \kappa \cdot \sum_{i=1}^{T} \frac{Q_i}{i},
    \end{align*}
    where $Q_i := \left\lvert \left\{ \ell~\middle\vert~\|\ell\|_{\infty} = i \right\} \right\rvert$ is the number of $\ell$ with $\ell_{\infty}$-distance exactly $i$. The last inequality holds since $T \ge 2n-1$.

    Using the trivial bound
    \begin{equation*}
        Q_i = (2i+1)^d - (2i-1)^d \le 2d \cdot (2i+1)^{d-1},
    \end{equation*}
    we have
    \begin{equation*}
        \sum_{i=1}^{T} \frac{Q_i}{i} \le 6d \cdot \sum_{i=1}^{2T+1} i^{d-2}.
    \end{equation*}
	Next, note that:
	\begin{itemize}
		\item When $d = 1$, this is upper bounded by $6 \cdot (\ln (2T+1) + 1) \le 30 \log (1/\eps)$. 
		\item When $d \ge 2$, this is upper bounded by $6d \cdot \frac{(2T+2)^{d-1}}{d-1} \le 12 \cdot \left(\frac{24\sqrt{d}}{\eps}\right)^{d-1}$.
	\end{itemize}
    In either case, by our choice of $\kappa$, it is easy to verify that
    \begin{equation*}
        \sum_{0 < \|\ell\|_{\infty} \le 2n-1} \left( \widehat{f_1}(\ell) - \widehat{f_2}(\ell) \right) \cdot \widehat{J_{d,n}}(\ell) \cdot \overline{\widehat{h}(\ell)} \le \kappa \cdot \sum_{i=1}^{T} \frac{Q_i}{i} < \frac{\eps}{6}\,,
    \end{equation*}
    which contradicts with  \Cref{eq:pf-thm-wasserstein-d-dim-ub:sum-ell-2n-ge-eps-6}. \qed

\subsection{Lower bounds}
\label{sec:wasserstein-lower-bounds}

We now turn to the proofs of \Cref{thm:low-dim-kappa-zero,thm:lb-infinite}. We also give a sharper lower bound matching \Cref{thm:multi-dim-Wass-ub-signals} up to constant factors when $d = 1$ in \Cref{subsubsec:1d}. 

\subsubsection{Proof of \Cref{thm:low-dim-kappa-zero}}
\label{sec:low-dim-kappa-zero}

We show in this section that even if we have the exact Fourier coefficients up to individual degree  $\sqrt{d}/(2\eps) - 1$ for a signal over the $d$-dimensional torus $\mathbb{T}^d$, we cannot reconstruct the signal within an $\eps$ Wasserstein distance.
In fact, this holds even for non-negative signals, i.e., for probability distributions. 

\WassLBKappaZero*

\begin{proof}
    Let $T' := \lfloor\sqrt{d}/(2\eps)\rfloor$. 
    Let $D_1$ be the uniform distribution over the set of all points of the form $(j_1/T', \dots, j_d/T')$ where $j_1, \dots, j_d \in \{0, 1, \dots, T' - 1\}$, and let $D_2$ be the uniform distribution over points of the form $(j_1/T' + 1/(2T'), \dots, j_d/T' + 1/(2T'))$ where $j_1, \dots, j_d \in \{0, 1, \dots, T' - 1\}$.

    We first argue that $\dW(D_1, D_2) \ge \eps$. This is true because for any $x \in D_1, y \in D_2$, 
    \[\dtor(x, y) = \sqrt{\dtor(x_1, y_1)^2 + \cdots + \dtor(x_d, y_d)^2} \ge \sqrt{d(1/(2T'))^2} \ge \eps.\]

    Below we fix any $\ell \in \mathbb{Z}^d$ such that $\|\ell\|_\infty \le \sqrt{d}/(2\eps) - 1$ and prove that $\widehat{D_1}(\ell) = \widehat{D_2}(\ell)$. Note that $\|\ell\|_\infty$ is an integer, so $\|\ell\|_\infty \le T' - 1$. 
    
    When $\ell = 0^d$, $\widehat{D_1}(\ell) = 1 = \widehat{D_2}(\ell)$ since $D_1$ and $D_2$ are distributions. 
    When $\ell \ne 0^d$, without loss of generality, suppose $\ell_1 \ne 0$. Since $|\ell_1| \le \|\ell\|_\infty \le T' - 1$, we have \[\sum_{j_1 = 0}^{T' - 1}e^{2\pi i\ell_1j_1/T'} = 0.\]
    Thus,
    \begin{align*}
    	\widehat{D_1}(\ell) &= \sum_{j \in \{0, 1, \dots, T' - 1\}^d}\frac1{T'^d}e^{2\pi i(\ell_1(j_1/T') + \cdots + \ell_d(j_d/T'))}\\
    &= \frac1{T'^d}\left(\sum_{j_1 = 0}^{T' - 1}e^{2\pi i\ell_1j_1/T'}\right) \cdots \left(\sum_{j_d = 0}^{T' - 1}e^{2\pi i\ell_dj_d/T'}\right)
    = 0\,,
    \end{align*}
    and similarly,
    \[\begin{split}
        \widehat{D_2}(\ell) &= \sum_{j \in \{0, 1, \dots, T' - 1\}^d}\frac1{T'^d}e^{2\pi i(\ell_1(j_1/T' + 1/(2T')) + \cdots + \ell_d(j_d/T' + 1/(2T')))}\\
        &= \frac{e^{\pi i(\ell_1 + \cdots + \ell_d)}}{T'^d}\left(\sum_{j_1 = 0}^{T' - 1}e^{2\pi i\ell_1j_1/T'}\right) \cdots \left(\sum_{j_d = 0}^{T' - 1}e^{2\pi i\ell_dj_d/T'}\right) = 0.
    \end{split}\]
    Therefore, for any $\ell \in \mathbb{Z}^d$ such that $\ell_\infty \le \sqrt{d}/(2\eps) - 1$, $\widehat{D_1}(\ell) = \widehat{D_2}(\ell)$.
\end{proof}

\begin{remark} \label{rem:LP}
    It is easy to see that the L\'evy-Prokhorov distance between the distributions $D_1$ and $D_2$ constructed in the above proof is also at least $\eps$. Thus, we cannot reconstruct a signal over $\mathbb{T}^d$ within an $\eps$ L\'evy-Prokorov distance even if we have the exact Fourier coefficients up to individual degree $\sqrt{d}/(2\eps) - 1$.
\end{remark}

\subsubsection{Proof of \Cref{thm:lb-infinite}}
\label{sec:low-dim-T-infinity} 

We show in this section that for a signal over the $d$-dimensional torus $\mathbb{T}^d$, even if we are given all (infinitely many) Fourier coefficients each to within $\pm \eps^{0.249d}$ additive error, we cannot reconstruct the signal within an $\eps$ Wasserstein distance. In fact, this holds for non-negative signals, i.e., probability distributions.

\WassLBInfiniteT*

\begin{proof}
    Let $\kappa = \eps^{0.249d}$,  
    $M = \frac12(8\eps)^{-0.5d}$, and
    $n = 4\eps^{-1}\sqrt d$.
    Let $x_1, \dots, x_M, y_1, \dots, y_M \in \mathbb{T}^d$ be $2M$ points that satisfy:
    \begin{enumerate}
        \item[(i)] For all $1 \le j, k \le M$, $\dtor(x_j, y_k) > 4\eps$.
        \item[(ii)] For all $\ell \in \mathbb{Z}^d$ such that $0 < \|\ell\|_\infty \le 2n$, we have 
        \[
            \left|\frac1M \sum_{j = 1}^Me^{2\pi i(\ell \cdot x_j)}\right| < \frac\kappa2 
            \qquad\text{and}\qquad 
            \left|\frac1M \sum_{k = 1}^Me^{2\pi i(\ell \cdot y_k)}\right| < \frac\kappa2\,.
        \]
    \end{enumerate}
    We show the existence of such $x_1, \dots, x_M, y_1, \dots, y_M$ in \Cref{claim:uniform-points-multi-dim-lb}.
    
    Now let $D'_1$ be the uniform distribution over $x_1, \dots, x_M$, and let $D_1 := D'_1 * J_{d, n}$ (so $D'_1$ is the uniform mixture of $M$ Jackson's kernel distributions centered at the points $x_1,\dots,x_M$). 
    Similarly, let $D'_2$ be the uniform distribution over $y_1, \dots, y_M$, and let $D_2 := D'_2 * J_{d, n}$.

    We first argue that $\dW(D_1, D_2) \ge \eps$. 
    By \Cref{lmm:jackson-expected-distance}, $\Ex_{\bx \sim J_{d, n}}[\dtor(\bx, 0^d)] \le \sqrt{d}/n$, so by Markov's inequality,
    \[\Prx_{\bx\sim J_{d, n}}\sbra{\dtor(\bx, 0^d) \le \eps} \ge 1 - \frac{\sqrt d}{\eps n} \ge 0.75.\]
    Recall that we defined $D_1'$ to be the uniform distribution over $x_1, \dots, x_M$ and $D_1 = D_1' * J_{d, n}$, so
    \[\Prx_{\bx \sim D_1}\sbra{\exists 1 \le j \le M\text{ s.t. }\dtor(\bx, x_j) \le \eps} \ge 0.75.\]
    Similarly,
    \[\Prx_{\by \sim D_2}\sbra{\exists 1 \le k \le M\text{ s.t. }\dtor(\by, y_k) \le \eps} \ge 0.75.\]
    Let $C$ be an arbitrary coupling of $D_1$ and $D_2$.  By a union bound,
    \[\Prx_{(\bx, \by) \sim C}\sbra{\exists 1 \le j, k \le M\text{ s.t. }\dtor(\bx, x_j) \le \eps\text{ and }\dtor(\by, y_k) \le \eps} \ge 0.5.\]
    Recall that $\dtor(x_j, y_k) > 4\eps$ for all $1 \le j, k \le M$, so 
    $\Prx_{(\bx, \by) \sim C}[\dtor(\bx, \by) > 2\eps] \ge 0.5$. This implies \[\dW(D_1, D_2) \ge \eps\] since $C$ is an arbitrary coupling of $D_1, D_2$.

    Now we show that for any $\ell \in \mathbb{Z}^d$, $\left|\widehat{D_1}(\ell) - \widehat{D_2}(\ell)\right| \le \kappa$. When $\|\ell\|_\infty \ge 2n$, by \Cref{lmm:jackson-fourier}, $\widehat{J_{d, n}}(\ell) = 0$, so $\widehat{D_1}(\ell) = 0 = \widehat{D_2}(\ell)$. When $\ell = 0^d$, $\widehat{D_1}(\ell) = 1 = \widehat{D_2}(\ell)$ since $D_1$ and $D_2$ are distributions. When $0 < \|\ell\|_\infty < 2n$, by \Cref{lmm:jackson-fourier}, $\left|\widehat{J_{d, n}}(\ell)\right| \le 1$, so 
    \[\left|\widehat{D_1}(\ell) - \widehat{D_2}(\ell)\right| = \left|\widehat{J_{d, n}}(\ell)\right| \cdot \left|\widehat{D_1'}(\ell) - \widehat{D_2'}(\ell)\right| \le \left|\frac1M\sum_{j = 1}^Me^{2\pi i(\ell \cdot x_j)} - \frac1M\sum_{k = 1}^Me^{2\pi i(\ell \cdot y_k)}\right| < \kappa,\]
    where the final uses (ii) and the triangle inequality.
\end{proof}

\begin{claim}\label{claim:uniform-points-multi-dim-lb}
    There exists $x_1, \dots, x_M, y_1, \dots, y_M \in \mathbb{T}^d$ that satisfy:
    \begin{enumerate}
        \item[(i)] For all $1 \le j, k \le M$, $\dtor(x_j, y_k) > 4\eps$.
        \item[(ii)] For all $\ell \in \mathbb{Z}^d$ such that $0 < \|\ell\|_\infty < 2n$, $\frac1M\left|\sum_{j = 1}^Me^{2\pi i(\ell \cdot x_j)}\right| < \frac\kappa2$ and $\frac1M\left|\sum_{k = 1}^Me^{2\pi i(\ell \cdot y_k)}\right| < \frac\kappa2$.
    \end{enumerate}
\end{claim}
\begin{proof}
The proof is by the probabilistic method: we uniformly sample $\bx_1, \dots, \bx_M, \by_1, \dots, \by_M$ from $\mathbb{T}^d$ and show that the outcome has the desired properties with positive probability.

    For any $1 \le j, k \le M$, we have
    \[\Pr[\dtor(\bx_j, \by_k) \le 4\eps] \le (8\eps)^d,\]
    so by union bound over all $1 \le j, k \le M$,
    \[\Pr[\bx_1, \dots, \bx_M, \by_1, \dots, \by_M\text { satisfy (i)}] \ge 1 - M^2(8\eps)^d \ge 0.75.\]

    For any $\ell \in \mathbb{Z}^d \setminus \{0^d\}$, without loss of generality assume $\ell_1 \ne 0$, then for any $1 \le j \le M$,
    \[\E[e^{2\pi i(\ell \cdot \bx_j)}] = \int_{\mathbb{T}^d}e^{2\pi i(\ell \cdot x)}dx = \int_{\mathbb{T}^d}e^{2\pi i\ell_1x_1}e^{2\pi i(\ell_2x_2 + \cdots + \ell_dx_d)}dx = 0\]
    since $\int_{\mathbb{T}}e^{2\pi i\ell_1x_1}dx_1 = 0$. Moreover, $|e^{2\pi i(\ell \cdot \bx_j)}| = 1$, so by Hoeffding's inequality,
    \begin{align*}
        \Pr\left[\left|\sum_{j = 1}^M\Re(e^{2\pi i(\ell \cdot \bx_j)})\right| \ge \frac{\kappa M}{2\sqrt2}\right] &\le 2e^{-\kappa^2 M/16}\\
        \Pr\left[\left|\sum_{j = 1}^M\Im(e^{2\pi i(\ell \cdot \bx_j)})\right| \ge \frac{\kappa M}{2\sqrt2}\right] &\le 2e^{-\kappa^2 M/16}
    \end{align*}
    and thus
    \[\Pr\left[\left|\sum_{j = 1}^Me^{2\pi i(\ell \cdot \bx_j)}\right| \ge \frac{\kappa M}{2}\right] \le 4e^{-\kappa^2 M/16}.\]
    Similarly, \[\Pr\left[\left|\sum_{j = 1}^Me^{2\pi i(\ell \cdot \by_k)}\right| \ge \frac{\kappa M}{2}\right] \le 4e^{-\kappa^2 M/16}.\]
    Therefore, by union bound over all $\ell \in \mathbb{Z}^d$ such that $0 < \|\ell\|_\infty < 2n$, 
    \[\Pr[\bx_1, \dots, \bx_M, \by_1, \dots, \by_M\text { satisfy (ii)}] \ge 1 - 8(4n - 1)^de^{-\kappa^2 M/16} \ge 0.75\]
    where the last inequality holds because $0 < \eps \leq \eps_0$ where $\eps_0$ is sufficiently small.
    It follows that
    \[\Pr[\bx_1, \dots, \bx_M, \by_1, \dots, \by_M\text { satisfy (i) and (ii)}] \ge 0.5,\]
    and the claim is proved.
\end{proof}

\begin{remark} \label{rem:LP2}
    By an argument similar to that for \Cref{thm:lb-infinite} (consider the union of $\eps$-radius balls centered at $x_j$), the L\'evy-Prokhorov distance between $D_1$ and $D_2$ is at least $\eps$. Thus, we cannot reconstruct a signal over $\mathbb{T}^d$ within an $\eps$ L\'evy-Prokhorov distance even if we have all Fourier coefficients each with $\eps^{0.249d}$ error.
\end{remark}




\subsubsection{A sharper one-dimensional lower bound}
\label{subsubsec:1d}

For the one-dimensional case, we can give a stronger lower bound than the one obtained by taking $d=1$ in \Cref{sec:low-dim-T-infinity}. This stronger lower bound implies that the upper bound obtained by taking $d=1$ in \Cref{sec:d-dim-ub-structural} is essentially the best possible. In more detail, we show that $\eps$-accurate Wasserstein reconstruction is impossible even if we are given $\pm 4 \eps$-accurate estimates of every Fourier coefficient, and even if the signal is promised to be non-negative (i.e.~it is some distribution $D$):
\begin{observation} \label{obs:1dlowerinfiniteT}
There are two distributions $D_1,D_2$ over the 1-dimensional torus $\mathbb{T}$ with the following properties:

\begin{itemize}
    \item $\abs{\widehat{D_1}(\ell) -\widehat{D_2}(\ell)} \leq 4 \eps$ for all $\ell \in \Z$, but
    \item $\dW(D_1,D_2) \geq \eps.$
\end{itemize}
\end{observation}

\begin{proof}
The construction is extremely simple: we take $D_1 := \delta_0$, i.e.~the distribution $D_1$ is a point mass at 0, and we take $D_2 = (1-2\eps)\cdot \delta_0 + 2\eps \cdot  \delta_{1/2}$, i.e.~a draw from $D_1$ outputs 0 with probability $1-2\eps$ and outputs $1/2$ with the remaining $2\eps$ probability.

It is clear that $\dW(D_1,D_2) = \eps,$ and we have
\[
\widehat{D_1}(\ell) = \int_{z \in [0,1)} \delta_0(z) e^{2 \pi i \ell z} dz = e^{2 \pi i \ell \cdot 0} = 1 \quad \quad \text{for all }\ell \in \Z,
\]
\begin{align*}
\widehat{D_2}(\ell) &= \int_{z \in [0,1)} \pbra{(1-2\eps)\cdot \delta_0(z) + 2\eps \cdot \delta_{1/2}(z)} e^{2 \pi i \ell z} dz\\
&= (1-2\eps) + (2\eps) e^{\pi i \ell}\\ 
&= 
\begin{cases}
1 & \text{~if $\ell$ is even}\\
1-4\eps & \text{~if $\ell$ is odd}
\end{cases},
\end{align*}
so $\abs{\widehat{D_1}(\ell) - \widehat{D_2}(\ell)} \leq 4 \eps$ for all $\ell \in \Z.$
\end{proof}



\subsection{An algorithm for Wasserstein reconstruction} \label{sec:wasserstein-algo}

In this section we show how the upper bounds for Wasserstein reconstruction that were established in \Cref{sec:multi-dim-Wasserstein} can be used to design an algorithm which, given the (noisy) low-frequency Fourier coefficients of some signal $f$, outputs a signal that is close to $f$ in Wasserstein distance. The idea is to discretize the space and use linear programming to find a hypothesis signal that has Fourier coefficients which are close to those of $f$; using the upper bounds established earlier,  such a hypothesis signal must be Wasserstein-close to $f$. 

We first formally state the theorem:

\begin{theorem} \label{thm:wasserstein-algorithm}
    There is an algorithm that satisfies the following: Fix any $d \ge 1$ and $0 < \eps < 1/8$, and let $T$ and $\kappa$ be as in \Cref{thm:multi-dim-Wass-ub}, i.e.,
     \[
    \text{let~}T = \left\lceil 6\sqrt{d}/\eps \right\rceil,
    \quad \text{and let~} 
    \kappa = \begin{cases}
        0.001 \eps / \log(1/\eps) & \text{~when~}d = 1\\
        \left(0.01 \eps / \sqrt{d}\right)^d & \text{~when~}d \ge 2
        \end{cases}.
    \]
    The algorithm takes as input $(2T+1)^d$ complex numbers $\{u_\ell~\vert~\ell \in \mathbb{Z}^d, \| \ell \|_{\infty} \le T\}$ and outputs a signal $g$ such that for any signal $f$ satisfying $| \widehat{f}(\ell) - u_\ell | \le \kappa/8$ for all $\| \ell \|_{\infty} \le T$, we have $\dW(f, g) \le 4\eps$. 
    The algorithm runs in time ${\pbra{d/\eps}^{O(d^2)}}$. In particular, for constant $d$, its running time is polynomial in the length of its input.
\end{theorem}



Regarding the representation of the hypothesis signal, the hypothesis is a ``Dirac comb'' signal supported on  $(d/\eps)^{O(d^2)}$ discrete points; the algorithm  outputs the location of the points of the Dirac comb and their associated (positive or negative) masses.

For convenience, throughout this section we write $[N]$ to denote $\{0, 1, \dots, N-1\}$.

\begin{proof}[Proof of \Cref{thm:wasserstein-algorithm}.]
    Let $K := \lceil 100d \cdot (2n)^{d+1}/\kappa \rceil$. We discretize $\mathbb{T}^d$ into $K^d$ points 
    \[
        (j_1/K, j_2/K, \dots, j_d/K)
    \]
    where $j$ ranges over $[K]^d$.
    
    The following notation will be useful: we use $a \in \mathbb{R}^{K^d}$ to represent a sequence of real numbers (which can be positive or negative),  $a = \left\{a_j~\middle\vert~j = (j_1, \dots, j_d) \text{~where~} 0 \le j_1, \dots, j_d < K \right\}$. We define the (possibly generalized) signal
    \begin{equation*}
        \Dirac_a(x) := \sum_{j \in [K]^d} a_j \cdot \delta_{(j_1/K, j_2/K, \dots, j_d/K)}(x).
    \end{equation*}
    When $\sum_j \lvert a_j \rvert = 1$, this is a bona fide (normalized) signal with mass $a_j$ at the point $(j_1/K, \dots, j_d/K)$. Let $\mathcal{U}_{[0,1/K)^d}$ denote the uniform distribution on $[0,1/K)^d$, and let
    \begin{equation*}
    \widetilde{\Dirac}_a(x) := \Dirac_a(x)~*~\mathcal{U}_{[0,1/K)^d}
    \end{equation*}
    where $*$ stands for convolution. In words, when $\sum_j \lvert a_j \rvert = 1$, then $\widetilde{\Dirac}_a$  is the signal with $\widetilde{\Dirac}_a(x) = K^d \cdot a_j$ in the rectangle $\left[\frac{j_1}{K}, \frac{j_1+1}{K}\right) \times \left[\frac{j_2}{K}, \frac{j_2+1}{K}\right) \times \dots \times \left[\frac{j_d}{K}, \frac{j_d+1}{K}\right)$.

    Given the input $\{u_\ell~\vert~\ell \in \mathbb{Z}^d \text{~s.t.~} \| \ell \|_{\infty} \le T\}$, let $f$ be the target signal with $ \widehat{f}(\ell) - u_\ell \le \kappa$ for all $\| \ell \|_{\infty} \le T$. Our algorithm has three steps which we describe and analyze below (the algorithm is also summarized in \Cref{alg:wasserstein-reconstruction}).
    
    \paragraph{Step 1: Jackson smoothing.} The first step is to apply a smoothing on $u$ and $f$ using Jackson's kernel (\Cref{def:jackson-kernel}). Let $n := \left\lceil \sqrt{d}/\eps \right\rceil$, and compute $\widetilde{u}_\ell := u_\ell \cdot \widehat{J_{d,n}}(\ell)$ for each $\ell \in \Z^d$ with $\|\ell\|_\infty \leq T$.
    
    Consider the function $\widetilde{f} := f * J_{d,n}$. This function has three properties that will be useful later:
    \begin{enumerate}
        \item \label{enum:tilde-f:wasserstein-distance} By \Cref{lmm:jackson-expected-distance}, $\dW(f, \widetilde{f}) \le \sqrt{d}/n \le \eps$.
        \item \label{enum:tilde-f:fourier} By \Cref{lmm:jackson-fourier}, $| \widehat{J_{d,n}}(\ell) | \le 1$. Combining with the assumption that $| \widehat{f}(\ell) - u_\ell | \le \kappa/8$ for all $\| \ell \|_{\infty} \le T$, we also have $| \widehat{\widetilde{f}}(\ell) - \widetilde{u}_\ell | \le \kappa/8$ for all $\| \ell \|_{\infty} \le T$.
        \item \label{enum:tilde-f:lipschitz} By \Cref{lmm:jackson-lipschitz} and the fact that $\int_{\mathbb{T}^d} \lvert f(x) \rvert dx = 1$, $\widetilde{f}$ is $(3\pi n^2(3n/2)^{d - 1}\sqrt d)$-Lipschitz.
        This uses the standard fact that for any signal $f$ and Lipschitz function $g$, the function $f * g$ inherits the Lipschitz constant of $g$ \cite{StackExchange2011PreservationLipschitzConvolution}. 
    \end{enumerate}

    \paragraph{Step 2: Linear programming.} The second step is to use linear programming to find some $a \in \mathbb{R}^{K^d}$ such that $\widetilde{\Dirac}_a$ is a signal that has Fourier coefficients close to $\widetilde{u}$. By \Cref{thm:multi-dim-Wass-ub}, this means that the Wasserstein distance between $\widetilde{\Dirac}_a$ and $\widetilde{f}$ is small.

    There is an obvious obstacle in the way of a naive formalization of the constraints on such an $a$ using linear programming: this is that the condition that $\widetilde{\Dirac}_a$ is a normalized signal, i.e., $\sum_{j \in [K]^d} \lvert a_j \rvert = 1$, is not linear. Instead, we will find some vector $a$ whose coordinates have total magnitude \emph{at most} $1$. We will deal with the possibility $\sum_{j \in [K]^d} \lvert a_j \rvert < 1$ later, in Step 3.

    For any $j \in [K]^d$, let
    \begin{equation*}
        c_{j,\ell} := \int_{\left[\frac{j_1}{K}, \frac{j_1+1}{K}\right) \times \left[\frac{j_2}{K}, \frac{j_2+1}{K}\right) \times \dots \times \left[\frac{j_d}{K}, \frac{j_d+1}{K}\right)} e^{2\pi i(\ell \cdot x)}~dx.
    \end{equation*}
    For any complex number $z \in \mathbb{C}$, denote the real part and the imaginary part of $z$ by $\Re(z)$ and $\Im(z)$, respectively. Consider the following linear program whose indeterminates are the coordinates $a_j$ of $a$:
\begin{align}
        \label{eq:wasserstein-ub:lp}
        \text{Find } &\{a_j\}_{j \in [K]^d} \\[3pt]
        \text{Subject to } 
        &\sum_{j \in [K]^d} |a_j| \le 1, \nonumber\\[3pt]
        &\bigg|\, K^d \sum_{j \in [K]^d} \Re(c_{j,\ell})\, a_j - \Re(\wt{u}_\ell) \,\bigg| 
           \le \frac{\kappa}{4} ,
           && \forall\, \ell \in \mathbb{Z}^d : \|\ell\|_\infty \le T, \nonumber\\[3pt]
        &\bigg|\, K^d \sum_{j \in [K]^d} \Im(c_{j,\ell})\, a_j - \Im(\wt{u}_\ell) \,\bigg| 
           \le \frac{\kappa}{4},
           && \forall\, \ell \in \mathbb{Z}^d : \|\ell\|_\infty \le T. \nonumber
    \end{align}
    In words, the second and third set of inequalities essentially say that the real parts and imaginary parts of $\widehat{\widetilde{\Dirac}_a}(\ell)$ are close to those of $\widetilde{u}_\ell$, respectively.

    We first show that the linear program (\ref{eq:wasserstein-ub:lp}) is satisfiable. Consider the following candidate solution $b \in \mathbb{R}^{K^d}$:
    \begin{align}
        b_j &:= K^{-d} \cdot \widetilde{f}(x) \text{~where $x$ minimizes $\left\lvert \widetilde{f}(x) \right\rvert$ within the rectangle} \nonumber \\
        & \ \ \ \ \left[\frac{j_1}{K}, \frac{j_1+1}{K}\right) \times \left[\frac{j_2}{K}, \frac{j_2+1}{K}\right) \times \dots \times \left[\frac{j_d}{K}, \frac{j_d+1}{K}\right). \label{eq:gatorade}
    \end{align}
    
    By definition,
    \begin{equation*}
        \sum_{j \in [K]^d} \lvert b_j \rvert = \int_{\mathbb{T}^d} \widetilde{\Dirac}_b(x)~dx \le \int_{\mathbb{T}^d} \left\lvert \widetilde{f}(x) \right\rvert~dx = 1,
    \end{equation*}
    satisfying the first constraint.

    We know by Property \ref{enum:tilde-f:lipschitz} that $\widetilde{f}$ is $(3\pi n^2(3n/2)^{d - 1}\sqrt d)$-Lipschitz, so we also have 
    \begin{equation*}
        \left\| \widetilde{\Dirac}_b - \widetilde{f} \right\|_{\infty} \le \frac{\sqrt{d}}{K} \cdot (3\pi n^2(3n/2)^{d - 1}\sqrt d) \le \kappa/8\,,
    \end{equation*}
    where we use that, on every rectangle (cf.~\Cref{eq:gatorade}), $\wt{\Dirac}_b$ agrees with $\wtf$ at some point; the Lipschitzness of  $\wtf$ and the bound $K$ on the rectangle's diameter then yields the desired uniform bound.
    Since for any (possibly generalized) signal $g$ and any $\ell \in \Z^d$ we have $|\widehat{g}(\ell)| \leq \|g\|_\infty$, for any $\ell \in \mathbb{Z}^d$ with $\|\ell\|_{\infty} \le T$ we have
    \begin{equation*}
        \left\lvert \widehat{\widetilde{\Dirac}_b}(\ell) - \widehat{\widetilde{f}}(\ell) \right\rvert \le \kappa/8.
    \end{equation*}
    We also know from Property \ref{enum:tilde-f:fourier} that
    \begin{equation*}
        | \widehat{\widetilde{f}}(\ell) - \widetilde{u}_\ell| \le \kappa/8, 
    \end{equation*}
    so by the triangle inequality we have
    \begin{equation*}
        \left\lvert \widehat{\widetilde{\Dirac}_b}(\ell) - \widetilde{u}_\ell \right\rvert \le \kappa/4.
    \end{equation*}
    Therefore, $b$ also satisfies the second and third inequalities of (\ref{eq:wasserstein-ub:lp}), and hence the linear program in \Cref{eq:wasserstein-ub:lp} is indeed feasible.

    Therefore, by solving \Cref{eq:wasserstein-ub:lp} with, say, the interior point method with error $\delta := \kappa/8$, in $\poly(K^d, \log(1/\delta))$ time we can find some $a^* \in \mathbb{R}^{K^d}$ satisfying the following linear program: 
    \begin{align}
        \label{eq:wasserstein-ub:lp-solved}
        \text{Find } &\{a^*_j\}_{j \in [K]^d} \\[3pt]
        \text{Subject to } 
        &\sum_{j \in [K]^d} |a^*_j| \le 1 + \delta, \nonumber\\[3pt]
        &\bigg|\, K^d \sum_{j \in [K]^d} \Re(c_{j,\ell})\, a^*_j - \Re(\wt{u}_\ell) \,\bigg| 
           \le \frac{\kappa}{4} + \delta,
           && \forall\, \ell \in \mathbb{Z}^d : \|\ell\|_\infty \le T, \nonumber\\[3pt]
        &\bigg|\, K^d \sum_{j \in [K]^d} \Im(c_{j,\ell})\, a^*_j - \Im(\wt{u}_\ell) \,\bigg| 
           \le \frac{\kappa}{4} + \delta,
           && \forall\, \ell \in \mathbb{Z}^d : \|\ell\|_\infty \le T. \nonumber
    \end{align}

    By the second and third sets of inequalities, for any $\ell \in \mathbb{Z}^d$ with $\|\ell\|_{\infty} \le T$, we have
    \begin{equation*}
        \left\lvert \widehat{\widetilde{\Dirac}_{a^*}}(\ell) - \widetilde{u}_\ell \right\rvert \le \kappa/2 + 2\delta \le 3\kappa/4.
    \end{equation*}
    Using Property \ref{enum:tilde-f:fourier} of $\widetilde{f}$ and the triangle inequality, we have
    \begin{equation*}
        \left\lvert \widehat{\widetilde{\Dirac}_{a^*}}(\ell) - \widehat{\widetilde{f}}(\ell) \right\rvert \le 3\kappa/4 + \kappa/8 < \kappa.
    \end{equation*}
    By \Cref{thm:multi-dim-Wass-ub}, this means that $\dW(\widetilde{\Dirac}_{a^*}, \widetilde{f}) \le \eps$.

    \paragraph{Step 3: Re-normalizing $a^*$.} Before re-normalizing, we first move to a Dirac comb function. Consider the function $\Dirac_{a^*}$; by the definition of $\widetilde{\Dirac}_{a^*}$, we have
    \begin{equation*}
        \dW (\Dirac_{a^*}, \widetilde{\Dirac}_{a^*}) \le \frac{\sqrt{d}}{K} \le \eps.
    \end{equation*}
    
    The next step is to re-normalize $a^*$ so that the absolute values sum up to $1$, so that the Dirac comb function becomes a signal. There are two cases here depending on $\gamma := \sum_{j \in [K]^d} \lvert a^*_j \rvert$.

    The first case is when $\gamma < 1$. Let $x_0$ and $x_1$ be two arbitrary points on $\mathbb{T}^d$ with Euclidean distance at most $\eps$ that are not on the grid $\{(j_1/K, j_2/K, \dots, j_d/K) \mid j \in [K]^d\}$ (for concreteness, take $x_0 := (1/(3K), 0, \dots, 0)$ and $x_1 := (2/(3K), 0, \dots, 0)$). Then $\|x_0 - x_1 \|_2 = 1/(3K) \le \eps$. Then define
    \begin{equation*}
        g(x) := \Dirac_{a^*}(x) + \frac{1-\gamma}{2} \delta_{x_0}(x) - \frac{1-\gamma}{2} \delta_{x_1}(x).
    \end{equation*}
    By definition, $\int_{\mathbb{T}^d} g(x) dx = 1$, and $\dW(g, \Dirac_{a^*}) \le (1-\gamma) \cdot \|x_0 - x_1 \|_2 \le \eps$.

    The second case is when $\gamma \ge 1$. By the first inequality of \Cref{eq:wasserstein-ub:lp-solved}, we have $\gamma \le 1+\delta = 1 + \kappa/8$. Define
    \begin{equation*}
        \widecheck{a^*_j} := \frac{1}{\gamma} \cdot a^*_j
    \end{equation*}
    and consider the signal.
    \begin{equation*}
        g(x) := \Dirac_{\widecheck{a^*}}(x)
    \end{equation*}
    By the definition of $\widecheck{a^*}$, it is easy to see that $\int_{\mathbb{T}^d} g(x) dx = 1$. Moreover,
    \begin{equation*}
        \dW(g, \Dirac_{a^*}) \le \sqrt{d} \cdot \sum_{j \in [K]^d} \lvert a^*_j - \widecheck{a^*_j} \rvert = \sqrt{d} \cdot (\gamma - 1) \le \eps.
    \end{equation*}

    \paragraph{Putting everything together.} In both of the above two cases, we get a signal $g$ with $\dW(g, \Dirac_{a^*}) \le \eps$. Putting everything together, we have
    \begin{equation*}
        \dW(f, g) \le \dW(f, \widetilde{f}) + \dW(\widetilde{f}, \widetilde{\Dirac}_{a^*}) + \dW(\widetilde{\Dirac}_{a^*}, \Dirac_{a^*}) + \dW(\Dirac_{a^*}, g) \le 4\eps.
    \end{equation*}
    The final algorithm is shown in \Cref{alg:wasserstein-reconstruction}.
\end{proof}

\begin{algorithm}[t]
    \caption{Wasserstein reconstruction given noisy low-frequency Fourier coefficients} \label{alg:wasserstein-reconstruction}

    \SetKwInOut{Input}{Input}
    \SetKwInOut{Output}{Output}

    \Input{$\{u_\ell : \ell \in \mathbb{Z}^d \text{~such that~} \| \ell \|_{\infty} \le T\}$}
    \Output{A signal $g$ s.t. $\dW(f,g) \le 4\eps$ for any $f$ with $\left\lvert \widehat{f}(\ell) - u_\ell \right\vert \le \kappa/8$ for all $\| \ell \|_{\infty} \le T$}

    \BlankLine
    \BlankLine

    Let $K \gets \lceil 100d \cdot (2n)^{d+1}/\kappa \rceil$, $n \gets \lceil \sqrt{d}/\eps \rceil$, $\delta \gets \kappa/8$\;
    Let $\widetilde{u}_\ell \gets u_\ell \cdot \widehat{J_{d,n}}(\ell)$\;
    Solve the LP in \Cref{eq:wasserstein-ub:lp} to obtain some $a^* \in \mathbb{R}^{K^d}$ satisfying \Cref{eq:wasserstein-ub:lp-solved}\;
    \eIf{$\gamma := \sum_{j \in [K]^d} \lvert a^*_j \rvert < 1$}{
        Output $g(x) := \Dirac_{a^*}(x) + \frac{1-\gamma}{2} \delta_{(1/(3K), 0, \dots, 0)}(x) - \frac{1-\gamma}{2} \delta_{(2/(3K), 0, \dots, 0)}(x).$
    }{
        Let $\widecheck{a^*_i} \gets a^*_i/\gamma$\;
        Output $g(x) := \Dirac_{\widecheck{a^*}}(x)$
    }
\end{algorithm}

\begin{remark}
    We note that when $f$ is a distribution instead of a signal, our algorithm can be modified to also output a distribution $g$. In this case, the linear program in \Cref{eq:wasserstein-ub:lp} should be replaced by
    \begin{align*}
    \text{Find } &\{a_j\}_{j \in [K]^d} \\[3pt]
    \text{Subject to }
    &\sum_{j \in [K]^d} a_j = 1, \\[3pt]
    &a_j \ge 0, && \forall\, j \in [K]^d, \\[3pt]
    &\bigg|\, K^d \sum_{j \in [K]^d} \Re(c_{j,\ell})\, a_j - \Re(\wt{u}_\ell) \,\bigg| \le \frac{\kappa}{4},
      && \forall\, \ell \in \mathbb{Z}^d : \|\ell\|_\infty \le T, \\[3pt]
    &\bigg|\, K^d \sum_{j \in [K]^d} \Im(c_{j,\ell})\, a_j - \Im(\wt{u}_\ell) \,\bigg| \le \frac{\kappa}{4},
      && \forall\, \ell \in \mathbb{Z}^d : \|\ell\|_\infty \le T.
    \end{align*}
    
    We can no longer proceed as we did formerly 
    in the $\gamma < 1$ case. Fortunately, we can  write down the constraint that $\sum_{j \in [K]^d} a_j = 1$ in a linear program, and a standard normalization suffices. The LP solution $a^*$ might have some negative coordinates, but we can simply clip all the negative $a^*_j$'s before doing the normalization.
\end{remark}

\end{document}